\documentclass[a4paper,11pt]{article}
\pdfoutput=1
\usepackage{a4wide}
\usepackage[UKenglish]{babel}
\usepackage[noadjust]{cite}
\usepackage{amsmath}
\usepackage{amssymb}
\usepackage{amsthm} 
\usepackage{bm}
\usepackage{bbm}
\usepackage{mathtools}
\usepackage{mathrsfs}
\usepackage{csquotes}
\usepackage{todonotes}
\usepackage[title]{appendix}
\usepackage{bbm}
\usepackage{xcolor}
\usepackage{blindtext}
\usepackage{enumerate}
\usepackage{manfnt}
\usepackage{multicol}
\usepackage{microtype}
\usepackage[ruled,linesnumbered,lined,boxed]{algorithm2e}
\usepackage{amsmath,amsfonts,amssymb}
\usepackage{graphicx,subcaption}

\usepackage{hyperref} 

\MakeRobust{\ref}

\makeatletter
\newcommand{\labeltext}[2]{%
  \@bsphack
  \csname phantomsection\endcsname 
  \def\@currentlabel{#1}{\label{#2}}%
  \@esphack
}
\makeatother

\newcommand{\trianglerightneq}{\mathrel{\ooalign{\raisebox{-0.5ex}{\reflectbox{\rotatebox{90}{$\nshortmid$}}}\cr$\triangleright$\cr}\mkern-3mu}}
\newcommand{\triangleleftneq}{\mathrel{\reflectbox{$\trianglerightneq$}}}

\usepackage[noframe]{showframe}
\usepackage{framed}
 {\endMakeFramed}
\definecolor{shadecolor}{gray}{0.88}

\usepackage{caption}
\usepackage{subcaption}
\usepackage[export]{adjustbox}
\usepackage{hyperref} 
\hypersetup{colorlinks=true,citecolor=blue,linkcolor=blue,urlcolor=blue}

\newcommand{\YES}{\textsc{Yes}}
\newcommand{\NO}{\textsc{No}}

\newcommand{\ignore}[1]{}
\newcommand{\A}{\mathbf{A}}
\newcommand{\B}{\mathbf{B}}
\newcommand{\C}{\mathbf{C}}
\newcommand{\G}{\mathbf{G}}
\newcommand{\HH}{\mathcal{H}}
\newcommand{\GG}{\mathcal{G}}

\newcommand{\K}{\mathbf{K}}

\newcommand{\W}{\mathbf{W}}
\newcommand{\X}{\mathbf{X}}
\newcommand{\Y}{\mathbf{Y}}

\newcommand{\Vset}{V}
\newcommand{\Eset}{E}

\newcommand{\Xsym}{{\X^{\operatorname{sym}}}}
\newcommand{\Ysym}{{\Y^{\operatorname{sym}}}}
\newcommand{\Amon}{{\A^{\operatorname{mon}}}}
\newcommand{\Bmon}{{\B^{\operatorname{mon}}}}

\newcommand{\N}{\mathbb{N}}
\newcommand{\R}{\mathbb{R}}

\renewcommand{\vec}[1]{\mathbf{#1}}
\newcommand{\ba}{\vec{a}}
\newcommand{\bb}{\vec{b}}
\newcommand{\bc}{\vec{c}}

\newcommand{\bg}{\vec{g}}

\newcommand{\bp}{\vec{p}}
\newcommand{\bq}{\vec{q}}

\newcommand{\bx}{\vec{x}}

\newcommand{\by}{\vec{y}}
\newcommand{\bw}{\vec{w}}
\newcommand{\be}{\vec{e}}

\newcommand{\bz}{\vec{z}}

\newcommand{\bomega}{{\bm{\omega}}}

\DeclareMathOperator{\sdr}{sdr}

\DeclareMathOperator{\Pol}{Pol}

\DeclareMathOperator{\girth}{gir}
\DeclareMathOperator{\ind}{ind}
\DeclareMathOperator{\chr}{chr}
\DeclareMathOperator{\dom}{dom}

\DeclareMathOperator{\Lcon}{LC}

\DeclareMathOperator{\PCSP}{PCSP}
\DeclareMathOperator{\CSP}{CSP}

\DeclareMathOperator{\ar}{ar}
\DeclareMathOperator{\supp}{supp}
\DeclareMathOperator{\set}{set}
\DeclareMathOperator{\adj}{\mathscr{A}}

\newcommand{\fbr}[2]
{
\ensuremath{\varphi^{#1}_{#2}}
}

\newcommand{\bone}{\mathbf{1}}  
\newcommand{\bzero}{\mathbf{0}} 

\newcommand\ang[1]{{\ensuremath\langle #1\rangle}}

\theoremstyle{plain}
\newtheorem{thm}{Theorem}
\newtheorem*{thm*}{Theorem}
\newtheorem{lem}[thm]{Lemma}
\newtheorem*{lem*}{Lemma}
\newtheorem{prop}[thm]{Proposition}
\newtheorem*{prop*}{Proposition}
\newtheorem{cor}[thm]{Corollary}
\newtheorem*{cor*}{Corollary}

\theoremstyle{definition}
\newtheorem{defn}[thm]{Definition}
\newtheorem*{defn*}{Definition}
\newtheorem{rem}[thm]{Remark}

\begin{document}

\title{The periodic structure of local consistency\thanks{This work was supported by UKRI EP/X024431/1. For the purpose of Open Access, the authors have applied a CC BY public copyright licence to any Author Accepted Manuscript version arising from this submission. All data is provided in full in the results section of this paper.}}

\author{Lorenzo Ciardo\\
University of Oxford\\
\texttt{lorenzo.ciardo@cs.ox.ac.uk}
\and
Stanislav \v{Z}ivn\'y\\
University of Oxford\\
\texttt{standa.zivny@cs.ox.ac.uk}
}

\date{\today}
\maketitle

\begin{abstract}
We connect the mixing behaviour of random walks over a graph to the power of the local-consistency algorithm for the solution of the corresponding constraint satisfaction problem (CSP). We extend this connection to arbitrary CSPs and their promise variant. In this way, we establish a linear-level (and, thus, optimal) lower bound against the local-consistency algorithm applied to the class of aperiodic promise CSPs.
The proof is based on a combination of the probabilistic method for random Erd\H{o}s-R\'enyi hypergraphs and a structural result on the number of fibers (i.e., long chains of hyperedges) in sparse hypergraphs of large girth.
As a corollary, we completely classify the power of local consistency for the approximate graph homomorphism problem by establishing that, in the nontrivial cases, the problem has linear width.
\end{abstract}

\section{Introduction}

The algebraic 
characterisation 
of the power of local consistency~\cite{Bulatov09:width,Barto14:jacm}
was one of the fundamental milestones of the algebraic theory of constraint satisfaction problems
(CSPs)~\cite{Jeavons97:closure,Bulatov05:classifying}, and a key step towards the positive resolution of Feder--Vardi's Dichotomy Conjecture~\cite{Feder98:monotone} obtained, independently, by
Bulatov~\cite{Bulatov17:focs} and Zhuk~\cite{Zhuk20:jacm}.
Fix a relational structure $\A$ (consisting of a finite domain $A$ and a finite
list of relations over $A$). The \emph{constraint satisfaction problem} parameterised
by $\A$ (in short, $\CSP(\A)$) is the following computational problem: Given an
\emph{instance} structure $\X$ that admits a homomorphism (i.e., a
relation-preserving map) to $\A$, find an explicit such
homomorphism. We denote the presence of a homomorphism from $\X$ to
$\A$ by the expression $\X\to\A$. Suitably choosing the \emph{template} $\A$
allows for casting as a CSP many diverse problems, including \emph{graph
colouring} (if $\A$ is a clique $\K_c$), \emph{graph homomorphism} (if $\A$ is
an undirected graph), \emph{linear equations} over a finite field (if $\A$
encodes the affine spaces over the field), and several variants of the
\emph{Boolean satisfiability problem} such as \textsc{Horn Sat}, \textsc{$k$-in-$n$ Sat}, and \textsc{Not-All-Equal Sat}. 

Enforcing local consistency is one of the---if not the---main algorithmic approaches developed
for solving CSPs.%
\footnote{Another approach is based around the idea of finding a
small generating basis of the solution
space (generalising Gaussian elimination)~\cite{Bulatov06:sicomp,Idziak10:siam}; there are also combinations of the two approaches, cf.~\cite{Bulatov19:ic,bgwz20,cz23sicomp:clap,Dalmau24:lics}.}
The $\kappa$-th level of the local-consistency algorithm searches a system of locally compatible homomorphisms from substructures of $\X$ of size up to
$\kappa$ to $\A$. We call such system a \emph{$\kappa$-strategy}.
The \emph{width} of $\CSP(\A)$ is then the minimum $\kappa$ for which a $\kappa$-strategy can always be extended to a global solution. 
If $\kappa$ is a constant, the $\kappa$-th level of the algorithm runs in time polynomial in the size of the instance. Hence, any CSP having \emph{bounded} width is in P.\footnote{For CSPs, the notion of bounded width admits various equivalent descriptions, in
terms of expressibility via Datalog programs~\cite{Kolaitis00:jcss}, pebble
games~\cite{Feder98:monotone}, tree-width duality~\cite{Bulatov08:dulaties},
and robust tractability~\cite{Barto16:sicomp}, as well as solvability via
different algorithmic techniques including linear programming~\cite{tz17:sicomp}
and semidefinite programming~\cite{tz18,AtseriasO19}.} 
Both
Bulatov's~\cite{Bulatov17:focs} 
and Zhuk's~\cite{Zhuk20:jacm}  algorithms use local-consistency checking as a subroutine. Indeed, the ability to capture the power of local consistency~\cite{Bulatov09:width,Barto14:jacm} via
tame congruence theory~\cite{Hobby88:structure}
was a
key step towards the resolution of the Dichotomy Conjecture. In particular,
in~\cite{Kozik15:au}, it was established that the class of bounded-width CSPs admits a natural characterisation in terms of universal-algebraic objects known as \emph{polymorphisms}---i.e., homomorphisms $f:\A^n\to\A$, where $\A^n$ is
the $n$-th direct power of $\A$.
More precisely, it was shown that
$\CSP(\A)$ has bounded width precisely when it admits \emph{weak near-unanimity} (WNU)
polymorphisms of all arities---where $f$ is a WNU operation when it
satisfies the identities
$f(a,b,b,\dots,b,b)=f(b,a,b,\dots,b,b)=\ldots=f(b,b,b,\dots,a,b)=f(b,b,b,\dots,b,a)$.

Despite its success in the realm of CSPs, the algebraic description of the
power of local consistency 
has a
limitation---it does not extend to the \emph{promise} variant of CSPs. 
This can be formulated as follows: Fix two structures $\A$ and $\B$ such that $\A\to\B$; given some structure $\X$ such that $\X\to\A$ as input, find an explicit homomorphism from $\X$ to $\B$.
Here, $\A$ and $\B$ should be interpreted as the \emph{strong} and the \emph{weak} versions of the template, respectively. Hence, the promise $\CSP$ parameterised by $\A$ and $\B$ (in short, $\PCSP(\A,\B)$) asks to find an explicit, weakly satisfying assignment for an instance that is promised to be strongly satisfiable; in this sense, it captures approximability of the CSP~\cite{AGH17,BG21,BBKO21}.
Many fundamental computational questions that are inexpressible as standard CSPs
admit a formulation as PCSPs. Primary examples are the \emph{approximate graph colouring problem} and its natural generalisation, the \emph{approximate graph homomorphism problem}. The former corresponds to letting $\A$ and $\B$ be two cliques $\K_c$ and $\K_d$, respectively, with $c\leq d$; the latter corresponds to letting $\A$ and $\B$ be two undirected graphs. 
The complexity of both these problems is notoriously open. This is in striking
contrast with their non-promise variants---the graph colouring and the graph
homomorphism problems---whose complexity was classified already in 1972 by
Karp~\cite{Karp72} and in 1990 by Hell--Ne{\v{s}}et{\v{r}}il~\cite{HellN90}, respectively.

Part of the algebraic machinery developed in the decades of research on CSPs can be transferred to PCSPs.
In particular, polymorphisms have a natural promise version (homomorphisms $\A^n\to\B$) and, just like for CSPs, they completely determine the complexity of PCSPs, as established in~\cite{BG21,BBKO21}.
Yet, the CSP tools appear to lose a substantial part of their sharpness when they are adapted to the promise world.
The ultimate reason for this weakening is that the common, universal-algebraic
core underpinning many distinct CSP phenomena is less structured for PCSPs. In
particular, PCSP polymorphisms do not compose, which makes it
impossible to obtain new polymorphic identities from old ones. A clear example
of this phenomenon is the break of the ``WNU -- local consistency'' bond: While the local-consistency algorithm can be  applied to solve PCSPs and, thus, the notion of width is well defined in the promise setting, 
the presence of WNU polymorphisms of all arities does not guarantee bounded width for PCSPs, as established in~\cite{Atserias22:soda}. For example,
the PCSP parameterised by the structures encoding \textsc{$2$-in-$4$ Sat} and \textsc{Not-All-Equal Sat}
has WNU polymorphisms of all arities although
its width is unbounded---and, in fact, linear; see later. Unless P$=$NP, WNU polymorphisms do not even guarantee polynomial-time tractability (via any algorithm): For instance, $\PCSP(\K_c,\K_{2c})$ has WNUs of all arities; yet, it was shown in~\cite{KOWZ22} that the problem is NP-hard for $c\geq 6$.  
In particular, this makes it infeasible to investigate the width of promise problems such as approximate graph homomorphism using the polymorphic tools currently available.

\paragraph{Contributions}
We adopt a different, \emph{structural} approach to study the
power of local consistency.
To illustrate it, we start with a simple example---$\CSP(\C_n)$, where $\C_n$ is the undirected cycle of length $n$. In this case, the parity of $n$ determines the width of the problem: It is well known and not hard to check that $\CSP(\C_n)$ has width $3$ if $n$ is even, while it has unbounded width if $n$ is odd.
Can we link these two different width behaviours to a separating \emph{structural} property depending on the parity of the cycle length?

Consider the setting of a random walker starting from some vertex $v_0$ of $\C_n$ and moving, at each discrete time step, to a vertex adjacent to her current location, picking the left or the right one with probability $\frac{1}{2}$, independently. After a number $\tau$ of time steps, suppose that the walker ends up on vertex $v_\tau$. How much information on the initial position $v_0$ can she deduce from knowing the final position $v_\tau$ as well as the elapsed time $\tau$? If $n$ is odd, all information gets lost after a sufficiently long time. Indeed, in this case, any vertex can be reached from any other vertex in precisely $\tau$ steps, for any $\tau\geq n-1$. On the other hand, if $n$ is even, while the initial vertex cannot be determined with certainty, its parity may be recovered by summing $v_\tau$ and $\tau$ (mod $2$). 
This behaviour of random walks is known as \emph{periodicity}:\footnote{See Remark~\ref{rem_markov_chains} for a further discussion on this point.}  Even cycles are periodic, odd cycles are aperiodic.
Thus, in the very special case of CSPs over undirected cycles,
the different width of odd and even cycles corresponds to the periodicity property of the corresponding random walk.

The main conceptual contribution of this work is establishing that the power of local consistency is deeply linked to the periodic behaviour of the random walk described above. Moreover, this connection holds in a surprisingly broader context---well beyond the state-of-the-art knowledge on the complexity of the corresponding computational problems. 
For undirected graphs, aperiodicity corresponds 
to the existence of some integer $\tau$ (the \emph{mixing time}) such that any two vertices of the graph are connected by a walk of length $\tau$.\footnote{In the literature on Markov chains, there exist several different formalisations of the concept of ``mixing time''. One primary example is \emph{Kemeny's constant}---the expected number of time steps needed to connect two vertices, independently sampled according to the stationary distribution of the chain~\cite{kemeny1969finite,levene2002kemeny}. The definition used in this work is a combinatorial version of the latter (generalised to arbitrary relational structures), in that it minimises the number of steps needed for connecting \emph{every} pair of vertices; see Definition~\ref{defn_aperiodicity_monic}.} In the case of directed graphs (digraphs), the definition is similar, the only difference being that the presence of both a \emph{directed} and an \emph{alternating} walks of length $\tau$ is required. In both cases, aperiodicity can be alternatively, naturally 
expressed in terms of the matrix-theoretic properties of \emph{primitivity} and \emph{irreducibility} (or \emph{ergodicity}) applied to the corresponding adjacency matrices.
Considering walks (and orientations thereof) in hypergraphs, the notion of aperiodicity is naturally extended to arbitrary relational structures---thus covering the whole class of CSPs. Lastly, 
and most significantly, 
\emph{we lift the link between aperiodicity and width to the promise regime}. Hence, unlike the WNU polymorphic description, this approach is applicable to PCSPs.

We have seen that a (P)CSP has bounded width if it is solved by a constant level of the local-consistency algorithm---in which case, the running time of the algorithm is polynomial in the size of the instance. (P)CSPs of \emph{linear} width lie
at the opposite end of the width spectrum: Not even $\epsilon n$ levels of the
local-consistency algorithm are enough to solve these
problems (where $n$ is the size of the
instance and $\epsilon$ is a positive constant).\footnote{\label{footnote_bounded_equals_nonlinear_for_CSPs}In
particular, every PCSP of linear width has unbounded width. For non-promise
CSPs, the converse also holds: Every CSP of unbounded width has linear width~\cite{AtseriasO19}. For PCSPs, the analogous statement is not known to hold.} 
Classic examples of CSPs of linear width are systems of linear
equations~\cite{Feder98:monotone,AtseriasO19} (see also~\cite{Grigoriev01:tcs,Tulsiani09:stoc,Chan15:jacm} for other lower bounds on systems of linear equations). 
Our main result identifies aperiodicity as a sufficient condition for having
linear---and, thus, unbounded---width.

\begin{thm}
\label{thm_aperiodic_quasi_linear_general_case}
    Let $\A\to\B$ be relational structures such that $\A$ is aperiodic and $\B$
    is loopless.\footnote{A structure $\B$ contains a \emph{loop} if there
    exists an element $b\in B$ such that every relation of $\B$ contains the
    constant tuple $(b,\dots,b)$ of suitable arity. Otherwise, we say that
    $\B$ is \emph{loopless}.} Then, $\PCSP(\A,\B)$ has linear width.
\end{thm}
Our result holds both for the \emph{search} version of $\PCSP(\A,\B)$ defined above and for its \emph{decision} version---distinguish whether an instance $\X$ is homomorphic to $\A$ or is not even homomorphic to $\B$.
Observe that $n$ levels of the local-consistency algorithm always suffice for instances of size $n$. Hence, the linear-level lower bound in Theorem~\ref{thm_aperiodic_quasi_linear_general_case} is tight (up to optimising the constant $\epsilon$).
An immediate consequence of the result above is the complete characterisation of the power of local consistency for the approximate graph homomorphism problem.

\begin{cor}
\label{cor_approximate_graph_homomorphism_width}
    Let $\A\to\B$ be undirected graphs, and suppose that $\A$ is non-bipartite and $\B$ is loopless. Then, $\PCSP(\A,\B)$ has linear width.\footnote{
It is well known that $\PCSP(\A,\B)$ has bounded width if $\A$ is a bipartite graph or $\B$ has a loop~\cite{Feder98:monotone}. Note that, in the statement of Corollary~\ref{cor_approximate_graph_homomorphism_width}, we are implicitly identifying an undirected graph with a symmetric digraph containing both orientations of each edge.}
\end{cor}

\paragraph{Related work}
Our work is inspired by the recent paper~\cite{Atserias22:soda} by Atserias--Dalmau, which gives a sufficient condition for a PCSP to have linear width.\footnote{While the result in~\cite{Atserias22:soda} is stated in terms of \emph{sublinear}, as opposed to linear, width, the proof yields a linear-level, optimal lower bound.}
In the language of the current paper,
that condition requires that any pair of vertices should be connected by oriented walks of length $2$; thus, it corresponds to the case $\tau=2$ of our aperiodicity notion.
Atserias--Dalmau's condition is satisfied, in particular, by cliques, thus resulting in a linear-level lower bound against local consistency for the solution of approximate graph \emph{colouring}. Extending the lower bound to approximate graph \emph{homomorphism} requires dealing with graphs having arbitrarily high mixing time and, thus, it needs the full power of our Theorem~\ref{thm_aperiodic_quasi_linear_general_case}. 
Some of the machinery of~\cite{Atserias22:soda} is used in this paper as a black box (in particular, a result on the sparsity of random hypergraphs obtained via Chernoff bounds). However, switching from $\tau=2$ to arbitrarily high $\tau$ requires rather different proof ideas. In particular, the concept of \emph{$\tau$-fibrosity} that we introduce in this work to study the number of long, isolated chains of hyperedges in sparse hypergraphs turns out to behave differently for $\tau\leq 3$ and for $\tau\geq 4$: In the latter case, large fibrosity is not implied by sparsity alone, and bounds on the length of the shortest cycles in hypergraphs need to be established via a probabilistic analysis.

The occurrence of long paths and cycles in sparse random hypergraphs and the corresponding phase transitions have recently been investigated in various different settings via probabilistic-combinatorial methods, see~\cite{cooley2021longest,gerbner2024stability,furedi2016stability,furedi2018stability}.
A recent line of work established lower bounds against relaxations of approximate graph colouring~\cite{cz23soda:aip,cz23stoc:ba} and homomorphism~\cite{CiardoZ24} based on linear and semidefinite programming as well as linear-equation solvers. NP-hardness of the former problem was proved in~\cite{GS20:icalp} conditional to Khot's $d$-to-$1$ Conjecture~\cite{Khot02stoc}. 
Brakensiek--Guruswami conjectured that approximate graph homomorphism should be tractable in polynomial time if $\A$ is bipartite or $\B$ has a loop, and NP-hard otherwise~\cite{BG21}. Therefore, our Corollary~\ref{cor_approximate_graph_homomorphism_width} is consistent with---and supports---the conjectured complexity dichotomy for this problem.

\paragraph{Notation}
We let $\N$ denote the set of positive integer numbers. For $n\in\N$, we let $[n]=\{1,\dots,n\}$, while $[0]=\emptyset$. We usually write tuples in boldface, and their entries in normal font; for example, the $i$-th entry of a tuple $\ba$ shall be denoted by $a_i$. Also, we denote by $\set(\ba)$ the set of all entries of $\ba$.
Given tuples $\ba,\bb,\bc,\dots$, we let $\ang{\ba,\bb,\bc,\dots}$ denote their concatenation.
Let $\sigma$ be a \emph{signature}; i.e., a finite list of relation symbols $R$, each with an assigned \emph{arity} $\ar(R)\in\N$. A \emph{relational structure} $\A$ over the signature $\sigma$ consists of a finite, nonempty domain $A$ and a relation $R^\A\subseteq A^{\ar(R)}$ for each symbol $R\in\sigma$. We denote by $n_\A$ the size of the domain of $\A$; i.e., $n_\A=|A|$.
Given two structures $\A$ and $\A'$, both over the signature $\sigma$, a \emph{homomorphism} from $\A$ to $\A'$ is a function $f:A\to A'$ such that, for each $R\in\sigma$ and each $\ba\in R^\A$, $f(\ba)\in R^{\A'}$ (where $f(\ba)=(f(a_1),\dots,f(a_{\ar(R)}))$ is the entrywise application of $f$ to the entries of $\ba$).


\section{Overview of results and techniques}
To prove Theorem~\ref{thm_aperiodic_quasi_linear_general_case}, we follow the intuition that the power of the local-consistency algorithm corresponds to the amount of information retained in a random walk.\footnote{The informal description below should be compared to the \emph{pebble-game} interpretation of local consistency~\cite{AtseriasO19}.} Suppose that we have two structures $\X$ (the \emph{instance}) and $\A$ (the \emph{template}), and assume, for simplicity, that they are digraphs. The goal is to show that some high level $\kappa$ of the local-consistency algorithm applied to $\X$ and $\A$ accepts (in symbols, $\Lcon^\kappa(\X,\A)=\YES$\footnote{The formal definition of the algorithm is given in Section~\ref{sec_prelimns}.}) but $\X\not\to\B$. This would mean that $\X$ is a \emph{fooling instance} for the algorithm applied to the decision version of $\PCSP(\A,\B)$; as the latter reduces to the search version~\cite{BBKO21}, this would be sufficient to prove Theorem~\ref{thm_aperiodic_quasi_linear_general_case} for both versions.
Consider two random walkers, Xerxes and Alice, walking on $\X$ and on $\A$, respectively. Let $x_0,x_1,x_2,\dots$ and $a_0,a_1,a_2,\dots$ be the two sequences of vertices resulting from their walks. In order for the algorithm to accept, we want that the assignment $x_i\mapsto a_i$ should locally look like a homomorphism. As discussed in the Introduction, if $\A$ is aperiodic, Alice has a limited memory: She is not able to retain any local information on her walk after some number $\tau$ of steps. This can be exploited by Xerxes.
Suppose that $\X$ contains a path of length $\tau$ isolated from the rest of the structure. Xerxes enters the path at some time $t$ and reaches the other end at time $t+\tau$. By that time, Alice will have no memory of her location at time $t$. Hence, she will not be able to spot any inconsistency involving the vertices of the path. If \emph{many} of these deceptive paths exist in $\X$, chances are that Xerxes will always manage to make his walk locally consistent with Alice's. Now, the longer time $\tau$ is needed for Alice to lose here memory, the longer Xerxes' paths need to be, in order to be effective.
The risk, of course, is that if \emph{too many} long paths are required, it could become easy to build a global homomorphism from $\X$ to any other structure $\B$, thus preventing it from being a fooling instance.  

Our proof of Theorem~\ref{thm_aperiodic_quasi_linear_general_case} consists of four stages, discussed in the next four subsections of this overview. 
The first stage (\emph{aperiodicity}) concerns the template; the second and third stages (\emph{fibrosity} and \emph{sparsity}) concern the instance; the final stage (\emph{consistency}) concerns both the instance and the template. Further details and complete proofs are given in the body of the paper.
The last subsection of this overview contains a discussion on some consequences of our main result.

\subsection{Aperiodicity}
\label{subsec_aperiodicity}
First of all, we need to formally define aperiodicity for arbitrary relational structures. If the structure has a unique, binary, symmetric relation (i.e., it is a graph), aperiodicity means that any pair of vertices is connected by a walk of some fixed length.
If the symmetry constraint is relaxed (thus resulting in a digraph), we require the presence of both a directed and an alternating walks of fixed length. When lifting the notion of aperiodicity to arbitrary structures, it is necessary to extend the definition of a walk, and decide what it means for it to be ``directed'' or ``alternating''. 
We first consider the case that $\A$ is \emph{monic}---i.e., its signature consists of a unique symbol.
\begin{defn}[Aperiodicity, monic case]
\label{defn_aperiodicity_monic}
    Let $\A$ be a monic relational structure and let $R$ be the unique symbol in its signature. Given $\tau\in\N$, we let a \emph{$\tau$-pattern} $\lambda$ be a $\tau\times 2$ matrix having entries in $[\ar(R)]$ such that the two entries in each row are different. Fix a $\tau$-pattern $\lambda$, and denote its entries by $\lambda_{i,j}$. A \emph{$\lambda$-walk} in $\A$ consists of a list $a_0,\dots,a_\tau$ of vertices in $A$ and a list $\ba^{(1)},\dots,\ba^{(\tau)}$ of tuples in $R^\A$ satisfying the following requirement: For each $i\in [\tau]$, $a^{(i)}_{\lambda_{i,1}}=a_{i-1}$ and $a^{(i)}_{\lambda_{i,2}}=a_{i}$. In this case, we say that the $\lambda$-walk \emph{connects} $a_0$ to $a_\tau$.
    We say that $\A$ is \emph{aperiodic} if there exists $\tau\in\N$ such that, for each $a,b\in A$ and each $\tau$-pattern $\lambda$, there exists a $\lambda$-walk connecting $a$ to $b$.\footnote{\label{cumbersome_footnote}To avoid cumbersome assumptions in the statements of our results, it shall be convenient to consider a monic structure having an empty relation of arity $1$ non-aperiodic.} In this case, the \emph{mixing time} of $\A$ is the minimum $\tau\in\N$ for which the condition holds.
\end{defn}
In the definition above, a $\lambda$-walk should be viewed as an ``oriented'' walk, where the orientation is prescribed by the $\tau$-pattern $\lambda$.
If $\A$ is not monic, we turn it into a monic structure $\Amon$ by taking the product of all relations in $\A$. Formally, let $\sigma=\{R_1,\dots,R_\ell\}$ be the signature of $\A$, and consider the signature $\tilde\sigma$ having a unique symbol of arity $\ar(R_1)+\ldots+\ar(R_\ell)$.
We let $\Amon$ be the monic structure over the signature $\tilde\sigma$, whose domain is $A$ and whose unique relation is the set $\{\ang{\ba^{(1)},\dots,\ba^{(\ell)}}:\ba^{(i)}\in R_i^\A\;\forall i\in[\ell]\}$.

\begin{defn}[Aperiodicity, general case]
\label{defn_aperiodicity_general}
A relational structure $\A$ is \emph{aperiodic} if $\Amon$ is aperiodic; in this case, the \emph{mixing time} of $\A$ is defined to be the mixing time of $\Amon$.
\end{defn}
As mentioned in the Introduction, the condition for linear width in~\cite[Thm.~3]{Atserias22:soda} is equivalent to the structure being aperiodic with mixing time $2$.
If $\A$ is a digraph,
each row of a $\tau$-pattern $\lambda$ is either the pair $(1,2)$ (the ``forward'' direction) or the pair $(2,1)$ (the ``backward'' direction). Thus, in this case, Definition~\ref{defn_aperiodicity_monic} asks for the existence of some $\tau$ such that every two vertices $a,b$ are connected by a walk of length $\tau$, oriented in \emph{every} possible way. It turns out that considering two orientations only (``directed'' and ``alternating'') is enough---thus recovering the definition encountered in the Introduction for the digraph case. In fact, via matrix-theoretic results (in particular, a classical theorem by Wielandt on the maximum index of primitivity of primitive matrices), we can make this condition even weaker. This also allows upper bounding the mixing time of an aperiodic digraph.

\begin{thm}
\label{thm_aperiodic_digraphs_combinatorial_defn_directed_alternating}
Let $\A$ be a digraph. The following are equivalent:
\begin{enumerate}
\itemsep-.2em 
    \item[$(i)$] $\A$ is aperiodic;
    \item[$(ii)$] there exists $t\in\N$ such that, for each $a,b\in A$, there exist a directed walk and an alternating walk from $a$ to $b$, both of length $t$;
    \item[$(iii)$] for each $a,b\in A$, there exist a directed walk of length $n_\A^2-2n_\A+2$ and an alternating walk of some even length from $a$ to $b$.
\end{enumerate}
If any of the equivalent conditions above holds, the mixing time of $\A$ is at most $n_\A^4-2n_\A^3+2n_\A^2$.
\end{thm}
We proceed with the proof of the main theorem, whose statement, having formally introduced aperiodicity, is now well defined.

\subsection{Fibrosity}
\label{subsec_fibrosity}
As we have seen---at least at the intuitive level---the key to making an instance $\X$ accepted by the local-consistency algorithm applied to an aperiodic template $\A$ is to require the presence of long paths (which we call \emph{fibers}) in $\X$, isolated from the rest of the instance (except for the endpoints). In particular, we shall require linearly many long fibers in all portions of the instance of some bounded size.
As illustrated by the game of Xerxes and Alice,
the length of these fibers should be at least equal to the mixing time of $\A$.
If, in addition to being highly fibrous, the instance has a large chromatic number, it is a good candidate for fooling local consistency. Instead of working with relational structures, it shall be easier, for now, to deal with hypergraphs.
Our strategy for achieving the goal is to split it into two subtasks: First, we show that a hypergraph meets the necessary fibrosity requirements provided that it is sparse and has large girth. Then, we prove that sparse hypergraphs having large girth and chromatic number do exist. In this subsection, we discuss the first of the two subtasks.

For an integer $r\geq 2$, an \emph{$r$-uniform hypergraph} $\HH=(\Vset(\HH),\Eset(\HH))$ consists of a nonempty set $\Vset(\HH)$ of \emph{vertices} and a set $\Eset(\HH)$ of \emph{hyperedges}, where a hyperedge is a set of vertices of cardinality exactly $r$.
Henceforth, we shall let $r$ be fixed, and we shall refer to $r$-uniform hypergraphs simply as hypergraphs. We denote the number of vertices and the number of hyperedges of $\HH$ by $n_\HH$ and $m_\HH$, respectively.
Two hyperedges are \emph{adjacent} if they have a nonempty intersection; a vertex is \emph{incident} to a hyperedge if it belongs to it.
For $v\in\Vset(\HH)$, $\deg_\HH(v)$ is the \emph{degree} of $v$; i.e., the number of hyperedges of $\HH$ to which $v$ is incident.
A vertex is \emph{isolated} if its degree is $0$.
We say that a hyperedge $e$ is a \emph{link} if exactly two vertices in $e$ have degree $2$, and all of the other $r-2$ vertices have degree $1$. 
A set $f$ of distinct links is called a \emph{fiber} if the elements of $f$ can be labelled as $e_1,\dots,e_\tau$, for some $\tau\in\N$, in a way that $e_i$ is adjacent to $e_{i+1}$ for each $i\in[\tau-1]$.  
If we need to stress that $f$ has cardinality $\tau$, we refer to it as a \emph{$\tau$-fiber}. The \emph{$\tau$-fibrosity} of $\HH$, which we denote by $\fbr{\tau}{\HH}$, is the maximum cardinality of a set of mutually disjoint $\tau$-fibers in $\HH$.

How to enforce large fibrosity in a hypergraph? Intuitively, if a hypergraph is dense, it is unlikely to contain many long fibers. Hence, we shall require that the hypergraph should be sparse. More precisely, given a real number $\beta>1$, we say that $\HH$ is \emph{$\beta$-sparse} if $m_{\HH}<\frac{\beta}{r-1}n_{\HH}$. Furthermore, we say that $\HH$ is \emph{hereditarily $\beta$-sparse} if every 
subhypergraph of $\HH$ (including $\HH$ itself) is $\beta$-sparse.
Unfortunately, sparsity is not enough for guaranteeing large fibrosity, even in the binary case. For example, the $\tau$-fibrosity of a star (i.e., a tree having all but one vertex of degree $1$) is $0$ while, like all acyclic graphs, stars are  hereditarily $\beta$-sparse for any $\beta>1$. Given a hyperedge $e\in\Eset(\HH)$, we say that $e$ is \emph{pendent} if at most one vertex of $e$ has degree at least $2$; 
the \emph{pendency} of $\HH$, denoted by $\pi_\HH$, is the number of pendent hyperedges in $\HH$. 
In the star case, the obstacle to large fibrosity is the presence of (linearly) many pendent edges.
As we shall see, these do not affect acceptance by the local-consistency algorithm. This suggests that the quantity to be lower-bounded should be the sum of the $\tau$-fibrosity and the pendency. While sparsity is enough for guaranteeing large $\tau$-fibrosity + pendency if $\tau\leq 3$, for larger values of $\tau$ it turns out to be insufficient.\footnote{We point out that this issue does not occur in~\cite{Atserias22:soda}, which corresponds to the case $\tau=2$ of our result.} 
For example, the graph obtained as the disjoint union of many triangles is hereditarily $\beta$-sparse for any $\beta>1$, but it has no pendent edges nor $\tau$-fibers for any $\tau>3$---while it has linearly many mutually disjoint $1$-fibers, $2$-fibers, and $3$-fibers. In this specific case, the problem is generated by the presence of short cycles---which we now define in the hypergraph case, following~\cite{berge1989hypergraphs}. For $2\leq\ell\in\N$,
a \emph{(Berge) cycle} of length $\ell$ in $\HH$ is a sequence 
$
v_0,e_0,v_1,e_1,\dots,v_{\ell-1},e_{\ell-1},v_0$
such that $v_i\in \Vset(\HH)$ and $e_i\in\Eset(\HH)$ for each $i$, all $v_i$'s and $e_i$'s are distinct, and $v_i\in e_{i-1}\cap e_i$
(mod~$\ell$). The \emph{girth} of $\HH$ (in symbols, $\girth(\HH)$) is the smallest length of a cycle in $\HH$. 
Note that
two distinct hyperedges of a hypergraph of girth at least $3$ cannot intersect in more than one vertices, as this would create a cycle of length $2$. %

The main result of this subsection is that short cycles are in fact the only obstacle: Sparsity and large girth are enough to guarantee linear fibrosity + pendency.

\begin{thm}
\label{thm_sparse_implies_large_fibrosity_plus_pendency}
    For $\tau\in\N$ and $1<\beta\in\R$, let $\HH$ be a hereditarily $\beta$-sparse hypergraph of girth at least $\tau$ having no isolated vertices. Then,
    $\fbr{\tau}{\HH}+\pi_\HH
        >
        \left(\frac{1}{10r\tau}-\beta+1\right)n_\HH$.
\end{thm}
\noindent If $\beta<1+\frac{1}{10r\tau}$, the inequality above yields a linear lower bound on the sum of $\tau$-fibrosity and pendency. We stress that this condition on $\beta$
involves $r$ and $\tau$, which are fixed constants of the template, but is independent of the number of vertices in $\HH$. As will become clear later, this is the first of two crucial details that are needed in order to make the instance able to fool linear---as opposed to only constant---levels of the local-consistency algorithm.

In order to prove Theorem~\ref{thm_sparse_implies_large_fibrosity_plus_pendency}, it shall be useful to 
introduce a more expressive nomenclature for fibers.
Given a fiber $f$ in $\HH$, let $\HH_f$ be the subhypergraph of $\HH$ induced by the vertices contained in the hyperedges of $f$.
We say that $f$ is \emph{degenerate} if $\HH_f$ is a cycle; otherwise, we say that $f$ is \emph{non-degenerate}.
We say that a fiber $f$ is \emph{maximal} if it is not properly included in any other fiber. We denote by $\fbr{\max}{\HH}$ the number of maximal fibers of $\HH$. Also, we denote by $\lambda_\HH$ the number of links of $\HH$. Notice that two distinct maximal fibers are necessarily disjoint; otherwise, their union would yield a longer fiber. Hence,
the set of maximal fibers partitions the set of links.
The proof of Theorem~\ref{thm_sparse_implies_large_fibrosity_plus_pendency} is obtained by showing that hereditarily sparse hypergraphs have a small number of maximal fibers (provided that no fiber is degenerate), while having a large number of links, as stated in the two results below.

\begin{prop}
\label{prop_maximal_fibers_upper_bound}
    For $\beta>1$, let $\HH$ be a hereditarily $\beta$-sparse hypergraph all of whose fibers are non-degenerate. Then,
    $
        \fbr{\max}{\HH}<3(\beta-1)n_\HH+3\pi_\HH
    $.
\end{prop}

\begin{prop}
\label{prop_many_links_if_no_pendent_edges}
    For $\beta>1$, let $\HH$ be a hereditarily $\beta$-sparse hypergraph having no isolated vertices. Then,
    $
        \lambda_\HH>\left(\frac{1}{r}+6-6\beta\right)n_\HH-7\pi_\HH$.
    \end{prop}

\noindent The numbers of links, maximal fibers, and $\tau$-fibers is constrained by the following relation.

\begin{prop}
\label{prop_nice_property_fibrosity_max_fibrosity}
    For any hypergraph $\HH$ and any $\tau\in\N$,
$
        \fbr{\tau}{\HH}+\fbr{\max}{\HH}
        >
        \frac{\lambda_\HH}{\tau}$.
\end{prop}

\noindent Then, the claimed lower bound on $\tau$-fibrosity + pendency easily follows.

\begin{proof}[Proof of Theorem~\ref{thm_sparse_implies_large_fibrosity_plus_pendency}]
Observe that $\tau$-fibrosity and pendency are additive with respect to disjoint unions; i.e., if $\HH$ and $\HH'$ are disjoint hypergraphs, $\fbr{\tau}{\HH\cup\HH'}=\fbr{\tau}{\HH}+\fbr{\tau}{\HH'}$ and $\pi_{\HH\cup\HH'}=\pi_\HH+\pi_{\HH'}$.
Clearly, the same holds for the number $n_\HH$ of vertices. Hence, we can assume without loss of generality that $\HH$ is connected.
Suppose first that all fibers of $\HH$ are non-degenerate. The combination of Propositions~\ref{prop_maximal_fibers_upper_bound},~\ref{prop_many_links_if_no_pendent_edges}, and~\ref{prop_nice_property_fibrosity_max_fibrosity} yields
\begin{align*}
    \fbr{\tau}{\HH}
    &>\frac{\lambda_\HH}{\tau}-\fbr{\max}{\HH}
    >
    \left(\frac{1}{r\tau}+\frac{6}{\tau}-\frac{6\beta}{\tau}\right)n_\HH-\frac{7}{\tau}\pi_\HH-3(\beta-1)n_\HH-3\pi_\HH\\
    &=
    \left(\frac{1}{r\tau}+\frac{6}{\tau}-\frac{6\beta}{\tau}-3(\beta-1)\right)n_\HH-\left(\frac{7}{\tau}+3\right)\pi_\HH
    >
    \left(\frac{1}{r\tau}-10\beta+10\right)n_\HH-10\pi_\HH,    
\end{align*}
whence it follows that
\begin{align*}
    \fbr{\tau}{\HH}+\pi_\HH
    \geq
    \frac{1}{10}(\fbr{\tau}{\HH}+10\pi_\HH)
    >
    \left(\frac{1}{10r\tau}-\beta+1\right)n_\HH,
\end{align*}
    as required.
    Suppose now that at least one fiber of $\HH$ is degenerate. Since $\HH$ is connected and has girth at least $\tau$, this implies that $\HH$ is a cycle of length at least $\tau$. Hence, $\tau\leq m_\HH=\frac{n_\HH}{r-1}$. By decomposing the cycle into mutually disjoint $\tau$-fibers, we deduce that, in this case,
    \begin{align*}
        \fbr{\tau}{\HH}+\pi_\HH
        &=
        \fbr{\tau}{\HH}
        =
        \left\lfloor\frac{m_\HH}{\tau}\right\rfloor
        >
        \frac{m_\HH}{2\tau}
        =
        \frac{n_\HH}{2(r-1)\tau}
        \geq
        \frac{n_\HH}{10 r\tau},
    \end{align*}
    and the claimed inequality follows since $\beta>1$.
\end{proof}

\subsection{Sparsity}
\label{subsec_sparsity}
We now know how to translate hereditary sparsity and large girth into linear fibrosity + pendency---which, in turn, shall be the key to making the instance locally consistent. The next step is to show that there exist hereditarily sparse hypergraphs of large girth that, at the same time, are highly chromatic. This last property will guarantee that the corresponding instance is a $\NO$-instance (i.e., it is not homomorphic to $\B$), thus making it able to fool the local-consistency algorithm.

We will need a slightly more general version of sparsity. Given two real numbers $\beta>1$ and $\gamma>0$, we say that a hypergraph $\HH$ is \emph{$(\gamma,\beta)$-threshold-sparse} if every 
subhypergraph $\HH'$ of $\HH$ such that $m_{\HH'}\leq\gamma$ is $\beta$-sparse. Observe that $\HH$ is hereditarily $\beta$-sparse precisely when it is $(m_\HH,\beta)$-threshold-sparse. Also, the \emph{chromatic number} of $\HH$---in symbols, $\chr(\HH)$---is the smallest integer $c$ for which the vertices of $\HH$ can be partitioned into $c$ classes such that every hyperedge intersects at least two classes of the partition.
We shall prove the following result.

\begin{thm}
\label{thm_sparse_incomparability_girth}
Take two positive integer numbers $g$ and $h$ and a real number $\beta>1$. There exists a positive real number $\delta=\delta(g,h,\beta)$ and a positive integer number $n_0=n_0(g,h,\beta)$ such that, for each $n\geq n_0$, there exists a hypergraph $\HH$ with the following properties:
\begin{align*}
    \begin{array}{llll}
        \mbox{1}. \mbox{ $\HH$ has $n$ vertices;} \hspace{2cm}& 
        \mbox{3}. \mbox{ $\chr(\HH)\geq h$;}  \\
        \mbox{2}. \mbox{ $\girth(\HH)\geq g$;} &
        \mbox{4}. \mbox{ $\HH$ is $(\delta n,\beta)$-threshold-sparse.}
    \end{array}
\end{align*}
\end{thm}

The proof of Theorem~\ref{thm_sparse_incomparability_girth} is probabilistic.
Specifically, we sample a random ($r$-uniform) Erd\H{o}s-R\'enyi hypergraph with $n$ vertices and hyperedge probability $p$ carefully chosen. For suitable values of $p$, one can show that $\HH$, with a nonzero probability, has the following properties: $(i)$ it contains a small number (less than, say, $\frac{n}{2}$) of short cycles; $(ii)$ it has a small independence number (less than $\frac{n}{2h}$); and $(iii)$ it is threshold-sparse.
At this point, to conclude, one simply needs to break all short cycles and check that the resulting hypergraph meets all the requirements listed in the theorem.
The proof of $(i)$ and $(ii)$ goes along the lines of the classic result on the existence of hypergraphs with arbitrarily large girth and chromatic number (obtained in~\cite{erdHos1966chromatic} as an extension of the analogous result for graphs). To prove $(iii)$, we use as a black box a sparsity result for random hypergraphs established in~\cite{Atserias22:soda} via Chernoff bounds (see also~\cite{chvatal1988many}).

It should be observed that, while the minimum required girth and chromatic number of the hypergraph obtained via the probabilistic construction above are fixed constants and, thus, can be arbitrarily small compared to the number $n$ of vertices, the first parameter of the threshold-sparsity is required to be linear in $n$. This is the second crucial detail that allows our argument to push the lower bound of Theorem~\ref{thm_aperiodic_quasi_linear_general_case} up to a linear level of consistency.
Indeed, unfolding the definition of threshold-sparsity, part $4.$ of Theorem~\ref{thm_sparse_incomparability_girth} asks that every subhypergraph of $\HH$ with at most $\delta n$ hyperedges should be (hereditarily) $\beta$-sparse. This allows calling Theorem~\ref{thm_sparse_implies_large_fibrosity_plus_pendency} for subhypergraphs of linear size---precisely what is needed in order to obtain consistency up to a linear level.

\subsection{Consistency}
\label{subsec_consistency}
In this final stage of our proof of Theorem~\ref{thm_aperiodic_quasi_linear_general_case}, we study the interaction between the template property of Subsection~\ref{subsec_aperiodicity} and the
instance properties of Subsections~\ref{subsec_fibrosity} and~\ref{subsec_sparsity}.
It shall be useful to select a set of one or two distinguished vertices (we call this set a \emph{joint}) from any pendent hyperedge and any fiber of a hypergraph $\HH$.
If $e$ is a pendent hyperedge of $\HH$, a set $J=\{v\}$ is a joint of $e$ if $v\in e$ and $v$ has maximum degree among the vertices of $e$. 
Let now $f$ be a fiber of $\HH$, and recall that $\HH_f$ denotes the subhypergraph of $\HH$ induced by the vertices contained in the hyperedges of $f$.
If $f$ is non-degenerate, there exist precisely two vertices $u\neq v$ of $\HH_f$ that are incident to edges not in $f$.
In this case, the joint of $f$ is the set $J=\{u,v\}$. Finally, if $f$ is degenerate (i.e., if $\HH_f$ is a cycle), we let a joint of $f$ be any set $J=\{v\}$ where $v$ is a vertex of $\HH_f$ such that $\deg_{\HH}(v)=2$. Notice that joints are not uniquely determined, in general. 
Given a monic structure $\X$ whose symbol $R$ has arity $r\geq 2$, we say that $\X$ is \emph{oriented} if $|\set(\bx)|=r$ for each $\bx\in R^\X$, and $\set(\bx)=\set(\bx')$ implies $\bx=\bx'$ for each $\bx,\bx'\in R^\X$.
Given an oriented monic structure $\X$, we let $\Xsym$ be the ($r$-uniform) hypergraph having vertex set $X$ and hyperedge set $\{\set(\bx):\bx\in R^\X\}$.

The next two results argue that, if $\A$ is aperiodic, a partial homomorphism to $\A$ can always by extended to include a fiber or a pendent hyperedge, respectively.\footnote{In the terminology of~\cite{Atserias22:soda}, this means that fibers and pendent hyperedges yield \emph{boundary sets} (see also~\cite{molloy2007resolution}).}

\begin{prop}
\label{prop_extension_tau_fibers}
    Let $\X,\A$ be monic structures such that $\X$ is oriented and $\A$ is aperiodic with mixing time $\tau$, let $f$ be a $\tau$-fiber of $\Xsym$, let $J$ be a joint of $f$, and let $\Y$ be the substructure of $\X$ induced by $(X\setminus \bigcup_{e\in f}e)\cup J$. Then, any homomorphism $h:\Y\to\A$ can be extended to a homomorphism $h':\X\to\A$.
\end{prop}

\begin{prop}
\label{prop_extension_pendent_edges}
    Let $\X,\A$ be monic structures such that $\X$ is oriented and $\A$ is aperiodic, let $e$ be a pendent hyperedge for $\Xsym$, let $J$ be a joint of $e$, and let $\Y$ be the substructure of $\X$ induced by $(X\setminus e)\cup J$. Then, any homomorphism $h:\Y\to\A$ can be extended to a homomorphism $h':\X\to\A$.
\end{prop}

When $\A$ is aperiodic and $\X$ is sparse and has a large girth,
the two propositions above allow for showing that acceptance by the local-consistency algorithm can be derived from the presence of a suitably large consistency gap, which we define next (following~\cite{Atserias22:soda}).

\begin{defn}
\label{defn_consistency}
Let $\X$ and $\A$ be monic structures, let $\Y$ be a substructure of $\X$, and let $X'$ be a subset of $X$. We say that a function $f:X'\to A$ is \emph{consistent with} $\Y$ if there is a homomorphism $h:\Y\to\A$ such that $f$ and $h$ agree on $X'\cap Y$. 
Moreover, for ${\kappa}\in\N$ and $\gamma\in\R$, we say that the pair $(\X,\A)$ has a \emph{$({\kappa},\gamma)$-consistency gap} if, for each substructure $\W$ of $\X$ such that $n_\W\leq {\kappa}$ and each homomorphism $f:\W\to\A$, one of the following conditions holds:
\begin{itemize}
\itemsep-.2em 
    \item $f$ is consistent with every substructure of $\X$ whose relation contains at most $\gamma$ tuples, or
    \item $f$ is not consistent with some substructure of $\X$ whose relation contains at most $\frac{\gamma}{n_\A}$ tuples.
\end{itemize}
\end{defn}

\begin{thm}
\label{thm_consistent_means_bounded_width}
    Let $\X$ and $\A$ be monic structures, and take two numbers ${\kappa}\in\N$ and $\gamma\in\R$ such that $\gamma\geq n_\A$. Suppose that 
    \begin{itemize}\itemsep-.2em 
        \item 
    $\A$ is aperiodic with mixing time $\tau$;
        \item 
     $\X$ is oriented, and $\X^{\operatorname{sym}}$ has girth at least $\tau$ and is $(\gamma,\beta)$-threshold-sparse for some real number $1<\beta<1+\frac{1}{10r\tau}$;
        \item 
      $(\X,\A)$ has a $({\kappa},\gamma)$-consistency gap.
    \end{itemize}
    Then, $\Lcon^{\kappa}(\X,\A)=\YES$.
\end{thm}

To finalise the proof of Theorem~\ref{thm_aperiodic_quasi_linear_general_case}, we just need to put all the pieces together: We construct a sparse, highly chromatic hypergraph $\HH$ of large girth via Theorem~\ref{thm_sparse_incomparability_girth}; we use Theorem~\ref{thm_sparse_implies_large_fibrosity_plus_pendency} to argue that all subhypergraphs of $\HH$ (up to linear size) have linear fibrosity + pendency; we turn $\HH$ into an oriented instance $\X$ by choosing an orientation for each hyperedge; we show that the pair $(\X,\A)$ has a $(\lfloor\epsilon n\rfloor,\delta n)$-consistency gap; we deduce via Theorem~\ref{thm_consistent_means_bounded_width} that $\Lcon^{\lfloor\epsilon n\rfloor}(\X,\A)=\YES$; finally, we use the facts that $\HH$ has large chromatic number and $\B$ is loopless to conclude that $\X$ fools linear levels of the local-consistency algorithm.

\subsection{Epilogue}
\label{subsec_epilogue}

Theorem~\ref{thm_aperiodic_quasi_linear_general_case} gives a sufficient condition for a PCSP to have linear width. In fact, we can effortlessly strengthen it by using the fact that, if $\A'\to\A\to\B$ and $\PCSP(\A',\B)$ has linear width, the same is true for $\PCSP(\A,\B)$---as is easily derived from Lemma~\ref{minion_homo_preserves_width}. We obtain the following result.
\begin{cor}
\label{cor_to_main_theorem}
    Let $\A\to\B$ be relational structures, and suppose that $\B$ is loopless and there exists some aperiodic relational structure $\A'$ such that $\A'\to\A$. Then, $\PCSP(\A,\B)$ has linear width.
\end{cor}

\noindent As an immediate consequence, we can completely characterise the power of local consistency applied to the approximate graph homomorphism problem.
\begin{proof}[Proof of Corollary~\ref{cor_approximate_graph_homomorphism_width}]
    Since $\A$ is non-bipartite, it contains an odd undirected cycle $\C_p$; thus, $\C_p\to\A$. 
    The result then follows from Corollary~\ref{cor_to_main_theorem} by observing that every odd cycle is aperiodic.
\end{proof}

A natural question, at this point, is whether the sufficient condition for linear width in Corollary~\ref{cor_to_main_theorem} is also necessary. The answer turns out to be negative, even in the CSP setting. Indeed, it is known that there exists an oriented tree $\textbf{T}$ whose corresponding CSP has unbounded width~\cite{BodirskyBSW23}. By~\cite{AtseriasO19} (see Footnote~\ref{footnote_bounded_equals_nonlinear_for_CSPs}), $\CSP(\textbf{T})$ has in fact linear width. Suppose that there exists some aperiodic digraph $\A'$ such that $\A'\to\textbf{T}$. It is not hard to check (see Proposition~\ref{prop_monotonicity_aperiodicity}) that this implies the existence of some aperiodic induced substructure $\textbf{T}'$ of $\textbf{T}$. This is a contradiction, as $\textbf{T}'$ is an oriented forest and, thus, no directed walk in $\textbf{T}'$ connects a vertex to itself.
Another natural question is whether there exist PCSPs having \emph{intermediate} width; i.e., super-constant but nonlinear. In the two cases of non-promise CSPs and binary symmetric PCSPs no such problems exist, as follows from~\cite{AtseriasO19} and from Corollary~\ref{cor_approximate_graph_homomorphism_width}, respectively. Does the width dichotomy extend to the entire class of PCSPs?  

Finally, we observe that, in the binary case, aperiodicity is formulated in terms of the matrix-theoretic notion of \emph{primitivity} applied to the adjacency matrix of $\A$ (see Section~\ref{sec_aperiodicity}). In turn, the nature of primitivity is ultimately spectral: The Perron--Frobenius Theorem implies that an irreducible matrix is primitive precisely when it has a unique eigenvalue having maximum modulus~\cite{minc1988nonnegative}. 
In other words, the necessary condition for nonlinear width given by the contrapositive of Theorem~\ref{thm_aperiodic_quasi_linear_general_case} involves the requirement that the adjacency spectrum of $\A$ should collapse, 
in the sense that at least two eigenvalues should have the same modulus. Now, a trivial topological consideration shows that the collapse of
the spectrum of a square matrix is
a \emph{singularity}, in the sense that the spectrum of a random matrix does not collapse with high probability. Resolving a 30-year-old conjecture of Babai, it was recently shown by Tao and Vu~\cite{tao2017random} that the discrete counterpart of this fact also holds: The adjacency spectrum of a random 
graph 
does not collapse with high probability.\footnote{Babai's conjecture was motivated by his work with Grigoryev and Mount~\cite{babai1982isomorphism}, proving that the graph \emph{isomorphism}
problem is in P when restricted to the class of graphs having distinct adjacency eigenvalues. Interestingly, the same condition makes the corresponding \emph{homomorphism} problem consistency-hard, as follows from the current paper.}
This suggests that, on a high level, PCSPs of nonlinear width should be regarded as a singularity within the class of PCSPs. 
While making this statement mathematically formal goes beyond the scope of the current paper, we find this phenomenon worth of further consideration.

\section{The local-consistency algorithm}
\label{sec_prelimns}
In this preliminary section, we formally describe the local-consistency algorithm, following~\cite{KolaitisV08,BKW17}.
Let $\X$ and $\A$ be two relational structures having the same signature $\sigma$.
We say that a structure $\Y$ of signature $\sigma$ is a \emph{substructure} of $\X$ if $Y\subseteq X$ and $R^\Y\subseteq R^\X$ for each $R\in\sigma$. If, in particular, $R^\Y=R^\X\cap Y^{\ar(R)}$ for each $R\in\sigma$, we say that $\Y$ is an \emph{induced substructure} of $\X$. (We define \emph{subhypergraphs} and \emph{induced subhypergraphs} analogously.) A \emph{partial homomorphism} from $\X$ to $\A$ is either a homomorphism from some induced substructure of $\X$ to $\A$, or the empty mapping from $\emptyset$ to $A$.
Take an integer $\kappa\in\N$, and consider a nonempty family $\mathscr{F}$ of partial homomorphisms from $\X$ to $\A$. Given a function $f$, let $\dom(f)$ denote its domain.
We say that $\mathscr{F}$ is a \emph{$\kappa$-strategy} for $\X$ and $\A$ if it is closed under restrictions (i.e., for every $f\in\mathscr{F}$ and $Y\subseteq\dom(f)$, the restriction of $f$ to $Y$ is in $\mathscr{F}$) and has the extension property up to $\kappa$ (i.e., for every $f\in\mathscr{F}$ with $|\dom(f)|<\kappa$ and every $x\in X\setminus\dom(f)$, there exists some $f'\in\mathscr{F}$ that extends $f$ and has domain $\dom(f)\cup \{x\}$). 
Testing for the existence of a $\kappa$-strategy for $\X$ and $\A$ can be performed in time polynomial in $(n_\X+n_\A)^\kappa$ through the so-called \emph{local-consistency algorithm}, which starts with all partial homomorphisms from $\X$ to $\A$ with domain size at most $\kappa$, and iteratively discards those that do not satisfy the two conditions above, until a fixed point is reached. If the resulting family of partial homomorphisms is nonempty, it must be a $\kappa$-strategy; in this case, we write $\Lcon^{\kappa}(\X,\A)=\YES$. Otherwise, we are guaranteed that no $\kappa$-strategy exists, and we write $\Lcon^{\kappa}(\X,\A)=\NO$. It shall be convenient to let $\Lcon^{0}(\X,\A)=\YES$ for each $\X$ and $\A$.

Let now $\A$ and $\B$ be two relational structures such that $\A\to\B$. 
Informally, the width of $\PCSP(\A,\B)$ measures the power of the local-consistency algorithm for its solution.
More precisely, observe that, given an instance $\X$, if $h$ is a homomorphism from $\X$ to $\A$, restricting $h$ to all subsets of $X$ of size at most $\kappa$ yields a proper $\kappa$-strategy for $\X$ and $\A$. Hence, if $\X\to\A$, $\Lcon^\kappa(\X,\A)=\YES$.
If, on the other hand, $\X\to\B$ whenever $\Lcon^\kappa(\X,\A)=\YES$, the $\kappa$-consistency algorithm effectively solves
$\PCSP(\A,\B)$.
Given $n\in\N$, let $\kappa(n)$ be the minimum nonnegative integer such that $\Lcon^{\kappa(n)}(\X,\A)=\YES$ implies $\X\to\B$ for every instance $\X$ having $n$ vertices. (Note that the minimum is well defined as, if $\X$ has $n$ vertices, it always holds that $\Lcon^n(\X,\A)=\YES$ if and only if $\X\to\A$; hence, in particular, $\kappa(n)\leq n$ for any PCSP.)
The resulting function $\kappa:\N\to\N\cup\{0\}$ is called the \emph{width} of $\PCSP(\A,\B)$. 
Observe that $\kappa$ is necessarily a non-decreasing function. Indeed, letting $\X$ be an instance on $n$ vertices such that $\Lcon^{\kappa(n)-1}(\X,\A)=\YES$ and $\X\not\to\B$, adding to $\X$ any number $c$ of isolated vertices results in an instance $\X'$ on $n+c$ vertices that satisfies $\Lcon^{\kappa(n)-1}(\X',\A)=\YES$ and $\X'\not\to\B$, thus witnessing that $\kappa(n+c)\geq\kappa(n)$.
If $\PCSP(\A,\B)$ has \emph{bounded} width (meaning that $\kappa(n)\leq \kappa^*$ for some fixed $\kappa^*$ independent of $n$), the local-consistency algorithm certifies that the problem is solvable in polynomial time. At the opposite extreme, we say that $\PCSP(\A,\B)$ has \emph{linear} width if there exists a positive constant $\epsilon$ for which $\kappa(n)>\epsilon n$ for any $n\in\N$; i.e., for any $n\in\N$ there exists an instance $\X$ on $n$ vertices such that 
$\Lcon^{\lfloor\epsilon n\rfloor}(\X,\A)=\YES$ but $\X\not\to\B$.

\section{Aperiodicity}
\label{sec_aperiodicity}

Let $\A$ be a digraph and let $w$ consist of a list $a_0,\dots,a_t$ of vertices of $\A$ and a list $\ba^{(1)},\dots,\ba^{(t)}$ of directed edges of $\A$. If $\ba^{(i)}=(a_{i-1},a_i)$ for each $i\in [t]$, we say that $w$ is a \emph{directed} walk of length $t$ from $a_0$ to $a_t$; if  $\ba^{(i)}=(a_{i-1},a_i)$ for each odd $i\in [t]$ and $\ba^{(i)}=(a_{i},a_{i-1})$ for each even $i\in [t]$, we say that $w$ is an \emph{alternating} walk of length $t$ from $a_0$ to $a_t$.   
The goal of this section is to prove the following result.
\begin{thm*}
[Theorem~\ref{thm_aperiodic_digraphs_combinatorial_defn_directed_alternating} restated]
Let $\A$ be a digraph. The following are equivalent:
\begin{enumerate}
\itemsep-.2em 
    \item[$(i)$] $\A$ is aperiodic;
    \item[$(ii)$] there exists $t\in\N$ such that, for each $a,b\in A$, there exist a directed walk and an alternating walk from $a$ to $b$, both of length $t$;
    \item[$(iii)$] for each $a,b\in A$, there exist a directed walk of length $n_\A^2-2n_\A+2$ and an alternating walk of some even length from $a$ to $b$.
\end{enumerate}
If any of the equivalent conditions above holds, the mixing time of $\A$ is at most $n_\A^4-2n_\A^3+2n_\A^2$.
\end{thm*}
The conditions $(ii)$ and $(iii)$ of Theorem~\ref{thm_aperiodic_digraphs_combinatorial_defn_directed_alternating} can be conveniently reformulated in terms of the adjacency matrix $\adj(\A)$ of $\A$---as we shall see after introducing the necessary terminology.
By $\be_i$, we indicate the $i$-th standard unit vector (whose size shall be clear from the context).
We say that a matrix is \emph{positive} (resp. \emph{nonnegative}) if all of its entries are positive (resp. nonnegative) real numbers. A nonnegative square matrix $M$ is \emph{primitive} if $M^t$ is positive for some $t\in\N$, and it is \emph{irreducible} if for each pair $(i,j)$ of indices there exists $t\in\N$ such that $\be_i^\top M^t\be_j$ (i.e., the $(i,j)$-th entry of $M^t$) is positive. Clearly, any primitive matrix is irreducible, while the converse is false. See also Remark~\ref{rem_markov_chains} for a more detailed discussion on irreducible and primitive matrices.

For two vertices $a,b\in A$,
a directed walk of length $t$ from $a$ to $b$ exists precisely when $\be_a^\top\adj(\A)^t\be_b>0$. Moreover, if $t$ is even, an alternating walk of length $t$ from $a$ to $b$ exists precisely when $\be_a^\top(\adj(\A)\adj(\A)^\top)^{\frac{t}{2}}\be_b>0$. Note that multiplying a positive and a primitive matrices yields a primitive matrix; moreover, $M^{t+1}$ is positive if $M^t$ is positive.
Thus, condition $(ii)$ in Theorem~\ref{thm_aperiodic_digraphs_combinatorial_defn_directed_alternating} is equivalent to requiring that both $\adj(\A)$ and $\adj(\A)\adj(\A)^\top$ should be primitive matrices. Moreover, condition $(iii)$ is equivalent to requiring that $\adj(\A)$ should be primitive with \emph{index of primitivity}\footnote{I.e., the smallest $t$ for which the $t$-th power of the matrix is positive.} at most $n_\A^2-2n_\A+2$, and $\adj(\A)\adj(\A)^\top$ should be irreducible.

We shall make use of the following two matrix-theoretic results. The first is a classic theorem by Wielandt, giving a (sharp) upper bound on the index of primitivity of primitive matrices. The second is folklore; the proof can be found, for example, in \cite[Ch.~3]{berman1994nonnegative}.
\begin{thm}[\cite{wielandt1950unzerlegbare,schneider2002wielandt}]
\label{thm_wielandt}
An $n\times n$ nonnegative matrix is primitive if and only if its $(n^2-2n+2)$-th power is positive.
\end{thm}

\begin{thm}[\cite{berman1994nonnegative}]
\label{thm_irreducible_positive_diagonal_means_primitive}
    Any irreducible matrix with positive trace is primitive.
\end{thm}

By $\supp(M)$ we denote the \emph{support} of $M$; i.e., the set of pairs $(i,j)$ of indices such that $\be_i^\top M\be_j\neq 0$.
Given two matrices $M_1$ and $M_2$ of equal size, we write $M_1\trianglelefteq M_2$ (resp. $M_1\triangleleftneq M_2$, $M_1\sim M_2$) to indicate that $\supp(M_1)\subseteq\supp(M_2)$ (resp. $\supp(M_1)\subsetneq\supp(M_2)$, $\supp(M_1)=\supp(M_2)$). The following fact is trivially proved.
\begin{lem}
\label{lem_basic_support}
    Let $M_1,M_2,M_3,M_4$ be nonnegative square matrices, and suppose that $M_1\trianglelefteq M_2$ and $M_3\trianglelefteq M_4$. Then, $M_1M_3\;\trianglelefteq\; M_2M_4$.
\end{lem}

We denote by $\bzero$, $I$, and $J$ the all-zero vector, the identity matrix, and the all-one matrix, respectively; their sizes shall be clear from the context.
Given a nonnegative square matrix $M$, a nonnegative integer $t$, and a tuple $\bx\in\{-1,1\}^t$ (which we shall refer to as an \emph{indicator tuple}), we consider the matrix $M^{\bx}=\prod_{1=1}^t N_i$, where $N_i=M$ if $x_i=1$, and $N_i=M^\top$ if $x_i=-1$. If $t=0$ (and, thus, $\bx$ is the empty tuple), we define $M^\bx=I$. 
Recall that the concatenation of the tuples $\bx,\by,\bz,\dots$ is denoted by $\ang{\bx,\by,\bz,\dots}$.

\begin{lem}
\label{lem_once_complete_everything_is_complete}
Let $M$ be an irreducible matrix and let $\bx,\by,\bz$ be indicator tuples such that $M^\by\sim J$. Then,  $ M^{\ang{\bx,\by,\bz}}\sim J$.
\end{lem}
\begin{proof}
It easily follows from the definition of an irreducible matrix that 
no row or column of $M$ is the all-zero vector.
We deduce that $M^\bw\be_i\neq\bzero$ for any indicator tuple $\bw$ and any index $i$. Take now two indices $j,\ell$, and observe that
    \begin{align*}
        \be_{j}^\top M^{\ang{\bx,\by,\bz}}
        \be_{\ell}
        =
        \be_{j}^\top M^\bx M^\by M^\bz\be_{\ell}.
    \end{align*}
    Since both $\be_{j}^\top M^\bx$ and $M^\bz\be_{\ell}$ are nonzero vectors and, by assumption, $M^\by\sim J$, we deduce that $\be_{j}^\top M^{\ang{\bx,\by,\bz}}
        \be_{\ell}\neq 0$, which concludes the proof.
\end{proof}
We say that an indicator tuple is \emph{balanced} if the sum of its entries is $0$. Also, we indicate by $\bomega_i$ the tuple $(1,-1,1,-1,\dots)$ of length $i$.
Henceforth in this section, we fix a digraph $\A$ and let $M$ be its $n_\A\times n_\A$ adjacency matrix.

\begin{lem}
\label{lem_balanced_tuple_identity}
Let $\A$ be a digraph satisfying part $(iii)$ of Theorem~\ref{thm_aperiodic_digraphs_combinatorial_defn_directed_alternating}, and let $\bx$ be a balanced indicator tuple. Then, $I\trianglelefteq M^\bx$. 
\end{lem}
\begin{proof}
    We use induction over the length $\ell$ of $\bx$ (which must be even, since $\bx$ is balanced). 
    If $\ell=0$, $M^\bx=I$ and the result is clear. If $\ell=2$, $\bx$ is either $\bomega_2$ or $-\bomega_2$. Since $M$ is primitive, it is irreducible, thus it has no all-zero row or column. Hence, the diagonals of the matrices $M^{\bomega_2}=MM^\top$ and $M^{-\bomega_2}=M^\top M$ are strictly positive, which means that $I\trianglelefteq M^\bx$, as needed. For the inductive step, suppose that $\ell\geq 4$ and notice that $\bx$ can be written as $\bx=\ang{\bw,\bz,\tilde\bw}$ with $\bz$ balanced of length $2$ and $\ang{\bw,\tilde\bw}$ balanced of length $\ell-2$. Using Lemma~\ref{lem_basic_support} and the inductive hypothesis, we deduce that 
    \begin{align*}
        I
        \trianglelefteq
        M^{\ang{\bw,\tilde\bw}}
        =
        M^\bw M^{\tilde\bw}
        =
        M^\bw I M^{\tilde\bw}
        \trianglelefteq
        M^\bw M^\bz M^{\tilde\bw}
        =
        M^{\ang{\bw,\bz,\tilde\bw}}
        =
        M^\bx,
    \end{align*}
    as needed.
\end{proof}

\begin{lem}
\label{lem_if_balanced_not_total_support_increases}
Let $\A$ be a digraph satisfying part $(iii)$ of Theorem~\ref{thm_aperiodic_digraphs_combinatorial_defn_directed_alternating}, let $\bx,\by$ be indicator tuples, and suppose that $\by$ is balanced and has a positive length. Then, either $J \sim M^\bx$ or $M^\bx\triangleleftneq M^{\ang{\bx,\by}}$.
\end{lem}
\begin{proof}
If $J\sim M^\bx$ there is nothing to prove, so we assume that this is not the case. Note that combining Lemma~\ref{lem_basic_support} and Lemma~\ref{lem_balanced_tuple_identity} yields $M^\bx\trianglelefteq M^{\ang{\bx,\by}}$, so we only need to exclude that $M^\bx\sim M^{\ang{\bx,\by}}$. Assume, for the sake of contradiction, that $M^\bx\sim M^{\ang{\bx,\by}}$.
By assumption, 
the matrix $MM^\top$ is irreducible. Moreover, applying Lemma~\ref{lem_balanced_tuple_identity} to the balanced indicator tuple $\bomega_2$ yields $I\trianglelefteq MM^\top$. We can then invoke Theorem~\ref{thm_irreducible_positive_diagonal_means_primitive} to deduce that $MM^\top$ is in fact primitive and, thus, $J\sim M^{\bomega_t}$ for some $t\in\N$.
We now proceed by induction on the length $\ell$ of $\by$. 

If $\ell=2$, suppose first that $\by=\bomega_2$. Take an even integer $z$.
Since $\bomega_{z}$ is balanced, Lemma~\ref{lem_basic_support} and Lemma~\ref{lem_balanced_tuple_identity} imply that $M^\bx\trianglelefteq M^{\ang{\bx,\bomega_{z}}}$. Moreover, using the assumption that $M^\bx\sim M^{\ang{\bx,\bomega_2}}$, a repeated application of Lemma~\ref{lem_basic_support} yields $ M^{\ang{\bx,\bomega_{z}}}\trianglelefteq M^\bx$. It follows that $M^\bx\sim M^{\ang{\bx,\bomega_{z}}}$.
Choosing $z$ so that $z\geq t$ and applying Lemma~\ref{lem_once_complete_everything_is_complete} yields 
\begin{align*}
    J\sim M^{\ang{\bx,\bomega_{z}}}\sim M^\bx,
\end{align*}
a contradiction. Similarly, if $\by=-\bomega_2$, we find $M^\bx\sim M^{\ang{\bx,-1,\bomega_z}}$ for any odd $z$, and the same argument as above yields the contradiction $J\sim M^\bx$.

For the inductive step, suppose that $\ell\geq 4$.
We can
write $\by$ as the concatenation $\by=\ang{\bw,\bz,\tilde\bw}$, where $\bz$ is balanced and has length $2$, and $\ang{\bw,\tilde\bw}$ is balanced and has length $\ell-2>0$. Lemma~\ref{lem_basic_support} and Lemma~\ref{lem_balanced_tuple_identity} yield $M^{\ang{\bx,\bw,\tilde\bw}}\trianglelefteq M^{\ang{\bx,\bw,\bz,\tilde\bw}}=M^{\ang{\bx,\by}}$, while
the inductive hypothesis gives $M^\bx\triangleleftneq M^{\ang{\bx,\bw,\tilde\bw}}$. As a consequence,
\begin{align*}
M^\bx
\triangleleftneq 
M^{\ang{\bx,\bw,\tilde\bw}}
\trianglelefteq
M^{\ang{\bx,\by}}
\sim 
M^\bx,
\end{align*}
a contradiction.
\end{proof}

Given a tuple $\bx$ of length $t$ and two indices $i\leq j\in [t]$, we denote by $\bx_{[i:j]}$ the tuple $(x_i,x_{i+1},\dots,x_j)$.

\begin{proof}[Proof of Theorem~\ref{thm_aperiodic_digraphs_combinatorial_defn_directed_alternating}]
To prove the implication $(i)\Rightarrow(ii)$, suppose that $\A$ is aperiodic and let $\tau$ be its mixing time. Consider the $\tau$-pattern $\lambda$ consisting in the $\tau\times 2$ matrix whose first column is identically $1$ and whose second column is identically $2$. Observe that a $\lambda$-walk is precisely a directed walk of length $\tau$. Similarly, an alternating walk of length $\tau$ is the same as a $\tilde\lambda$-walk corresponding to the $\tau$-pattern
\begin{align*}
    \tilde\lambda
    =
    \begin{bmatrix}
        1&2\\
        2&1\\
        1&2\\
        2&1\\
        \vdots&\vdots
    \end{bmatrix}.
\end{align*}
Hence, $(ii)$ follows from $(i)$ by the definition of aperiodic structures.

Observe now that the condition $(ii)$ is equivalent to the fact that $M^t\sim J$ and $M^{\bomega_t}\sim J$ for some $t\in\N$. By Lemma~\ref{lem_once_complete_everything_is_complete}, $t$ can be assumed to be even.
Hence, the implication $(ii)\Rightarrow(iii)$ directly follows by applying Theorem~\ref{thm_wielandt} to $M$.

We now establish the implication $(iii)\Rightarrow(i)$. Let $q=n_\A^2-2n_\A+2$, and observe that $(iii)$ implies that $M^q\sim J$.
Let $\tau=n_\A^4-2n_\A^3+2n_\A^2$. We claim that $\A$ is aperiodic with mixing time at most $\tau$. Given a $\tau$-pattern $\lambda$ and two vertices $a,b\in A$, observe that a $\lambda$-walk connecting $a$ and $b$ exists if and only if the $(a,b)$-th entry of the matrix $M^\bx$ is positive, where $\bx$ is the indicator tuple obtained by taking the first column of $\lambda$ and replacing all occurrences of ``$2$'' with ``$-1$''. Hence, the claim is equivalent to the fact that $J\sim M^\bx$ for any indicator tuple $\bx$ of length $\tau$. Pick one such tuple $\bx$.
Define integer numbers $y_0=0$ and, for $i\in [\tau]$, $y_i=y_{i-1}+x_i$. Let $m$ (resp. $\tilde m$) be the maximum (resp. the minimum) of the $y_i's$, and define $\mu=m-\tilde m$. We now consider two cases.

Suppose first that $\mu\geq q$. Let $i,j\in \{0,\dots,\tau\}$ be such that $y_i=\tilde m$ and $y_j=m$. Without loss of generality, we can assume that $i<j$ (otherwise, we proceed analogously reasoning on $M^\top$ instead of $M$). Consider the subtuple $\tilde\bx=\bx_{[i:j]}$, and
observe that it can be written as
\begin{align*}
\tilde\bx=\ang{1,\bw_1,1,\bw_2,1,\bw_3,1,\dots,\bw_{\mu-1},1},    
\end{align*}
where $\bw_1,\dots,\bw_{\mu-1}$ are balanced indicator tuples.
Using Lemma~\ref{lem_basic_support} and Lemma~\ref{lem_balanced_tuple_identity}, we deduce that
\begin{align*}
M^\mu\trianglelefteq
MM^{\bw_1}MM^{\bw_2}MM^{\bw_3}M\dots M^{\bw_{\mu-1}}M
=
M^{\tilde\bx}.
\end{align*}
Since $\mu\geq q$ and $J\sim M^q$, we deduce from Lemma~\ref{lem_once_complete_everything_is_complete} that $J\sim M^\mu\sim M^{\tilde\bx}\sim M^\bx$, as required. 

Suppose now that $\mu<q$. Since $\tau=qn_\A^2$, by the pigeonhole principle, there exists some integer $\ell$ such that $\tilde m\leq\ell\leq m$ and, letting $K=\{i\in\{0,\dots,\tau\}:y_i=\ell\}$, $|K|\geq\frac{\tau+1}{q}>\frac{\tau}{q}=n_\A^2$. Label the elements of $K$ as $k_1,\dots,k_z$ in increasing order (where $z=|K|$), and consider the subtuples $\bw_0=\bx_{[1:k_1]}$, $\bw_i=\bx_{[k_i+1:k_{i+1}]}$ for $i\in [z-1]$, and $\bw_z=\bx_{[k_z+1:\tau]}$. (If $k_1=0$, we let $\bw_0$ be the empty tuple; similarly, if $k_z=\tau$, we let $\bw_z$ be the empty tuple). Observe that $\bx=\ang{\bw_0,\bw_1,\dots,\bw_{z-1},\bw_z}$. Moreover, for each $i\in [z-1]$, $\bw_i$ is a balanced indicator tuple of length at least $2$. If $J\sim M^\ang{\bw_0,\bw_1,\dots,\bw_i}$ for some $i\in \{0,\dots,z\}$, we can conclude through Lemma~\ref{lem_once_complete_everything_is_complete} that $J\sim M^\bx$, as needed. Otherwise, Lemma~\ref{lem_if_balanced_not_total_support_increases} yields 
\begin{align*}
    M^{\bw_0}
    \triangleleftneq
    M^{\ang{\bw_0,\bw_1}}
    \triangleleftneq
    M^{\ang{\bw_0,\bw_1,\bw_2}}
    \triangleleftneq
    \dots
    \triangleleftneq
    M^{\ang{\bw_0,\dots,\bw_{z-1}}}
    \triangleleftneq
    J,
\end{align*}
which is impossible as $z>n_\A^2$.
\end{proof}

The next three remarks contain some further observations on aperiodicity.

\begin{rem}
\label{rem_markov_chains}
As we have discussed in the Overview at the intuitive level, our proof of Theorem~\ref{thm_aperiodic_quasi_linear_general_case} works by exploiting the fact that \emph{randomly walking over an aperiodic $\A$ for a sufficiently long time makes Alice oblivious of her past locations}. Using standard theory of discrete-time Markov chains,\footnote{For example, see~\cite{kemeny1969finite}.} we can give this statement a more rigorous meaning. Suppose, for the sake of simplicity, that $\A$ is strongly connected. Alice starts her walk from a vertex $a$ sampled uniformaly at random among all vertices of $\A$. After a number $i$ of time steps, she ascertains her new location $b$ and tries to guess her starting point $a$ based on the known data $b$ and $i$. Clearly, the trivial strategy of choosing at random yields a winning chance of $\frac{1}{n_\A}$. Can she do better?
Let $T$ be the transition matrix of $\A$; i.e., the row-stochastic matrix $T=D^{-1}M$, where $D$ is the diagonal matrix of the outdegrees of the vertices and $M$ is the adjacency matrix. 
Since $\A$ is strongly connected, the Markov chain is ergodic, which means that there exists a unique stationary probability distribution $\bp$ satisfying $\bp^\top T=\bp^\top$.
The probability of Alice ending up in $b$ after $i$ time steps starting from $a$ is $\be_{a}^\top T^i\be_{b}$ (where $\be_{a}$ is the probability distribution assigning all the weight to $a$).
By Bayes' Theorem, to maximise her winning chances, Alice should guess a vertex that maximises the corresponding entry of the vector $T^i\be_{b}$---which is a deterministic vector, as Alice is aware of both $i$ and $b$.
Now, it is easy to show that $T$ is primitive if and only if $M$ is primitive. If this is the case, $T^i$ converges to the rank-one matrix $\bone\bp^\top$ for $i\to\infty$ (where $\bone$ is the all-one vector). 
Hence, $T^i\be_{b}$ converges to the constant vector $p_b\bone$---meaning that, after a large-enough number of steps, all choices will roughly yield, for Alice, a disappointing $\frac{1}{n_\A}$ winning chance. On the other hand, if the primitivity requirement is relaxed, the Markov chain is not guaranteed to converge to the stationary distribution. In fact, it follows from the Perron--Frobenius Theorem that there exist subsequences of the sequence $(T^i)_{i\in\N}$ converging to precisely $c$ different limits, where $c$ is the number of eigenvalues of $M$ having maximum modulus~\cite[\S~9.2]{hogben2013handbook}; note that $M$ is primitive precisely when $c=1$. In particular, $T^i\be_{b}$ may not converge to a constant vector, which makes it possible for Alice to guess correctly with probability strictly larger than $\frac{1}{n_\A}$.
\end{rem}

\begin{figure}[!htb]
    \centering
    \begin{minipage}{.4\textwidth}
        \centering
        \scalebox{1.6}{\begin{tikzpicture}[scale=.7]
%
\draw[thick,->,shorten >=3pt,shorten <=0pt,>=stealth] (0,0) -- (1,1.73);
\draw[thick,->,shorten >=3pt,shorten <=0pt,>=stealth] (1,1.73) -- (2,0);
\draw[thick,->,shorten >=3pt,shorten <=0pt,>=stealth] (2,0) arc (300:240:2);
\draw[thick,->,shorten >=3pt,shorten <=0pt,>=stealth] (0,0) arc (120:60:2);
%
%
\draw[fill]  (0,0) circle (0.1cm);
\draw[fill]  (2,0) circle (0.1cm);
\draw[fill]  (1,1.73) circle (0.1cm);
\end{tikzpicture}}
      \caption*{(a) The adjacency matrix $M$ of the digraph above satisfies $M^5\sim J$, so $M$ is primitive. However, $MM^\top$ is not irreducible.}
    \label{fig_left}
    \end{minipage}%
    \hspace{1cm}
     \begin{minipage}{.4\textwidth}
        \centering
        \scalebox{1.6}{\begin{tikzpicture}[scale=.7]
%
\draw[thick,->,shorten >=3pt,shorten <=0pt,>=stealth] (1,1.73) -- (0,0);
\draw[thick,->,shorten >=3pt,shorten <=0pt,>=stealth] (1,1.73) -- (2,0);
\draw[thick,->,shorten >=3pt,shorten <=0pt,>=stealth] (2,0) arc (300:240:2);
\draw[thick,->,shorten >=3pt,shorten <=0pt,>=stealth] (0,0) arc (120:60:2);
%
%
\draw[fill]  (0,0) circle (0.1cm);
\draw[fill]  (2,0) circle (0.1cm);
\draw[fill]  (1,1.73) circle (0.1cm);
\end{tikzpicture}}
       \caption*{(b) The adjacency matrix $M$ of the digraph above satisfies $(MM^\top)^2\sim J$, so $MM^\top$ is primitive. However, $M$ is not irreducible.}
    \label{fig_right}
    \end{minipage}%
    \caption{Digraphs from Remark~\ref{rem:examples}.}\label{fig:digraphs}
\end{figure}

\begin{rem}\label{rem:examples}
While, in the case of arbitrary structures, the definition of aperiodicity
  requires that any pair of vertices should be connected by a walk of some fixed
  length $\tau$ oriented in all possible ways (more precisely, by a
  $\lambda$-walk for all possible $\tau$-patterns $\lambda$),
  Theorem~\ref{thm_aperiodic_digraphs_combinatorial_defn_directed_alternating}
  shows that, in the binary case, it is enough to consider two orientations only---directed and alternating. The condition cannot be further simplified:
  Figure~\ref{fig:digraphs} illustrates two examples 
of digraphs having the property that directed (resp. alternating) walks of some fixed length connecting any pair of vertices exist, but there are pairs of vertices not connected by any alternating (resp. directed) walk.
\end{rem}

\begin{rem}
In order to obtain more intuition of the definition of aperiodicity in the non-binary case, we now give a simple example of an aperiodic structure capturing equations over groups.
Let $G$ be a finite group with at least two elements, and let $\textbf{G}$ be the relational structure with domain $G$, having a ternary relation $R_g^\G=\{(h_1,h_2,h_3):h_1h_2h_3=g\}$ for each $g\in G$. Consider the structure $\G^{\operatorname{mon}}$, take $g,g'\in G$, and let $\lambda=(\lambda_1,\lambda_2)$ be a $1$-pattern, where $\lambda_1\neq\lambda_2\in[3\cdot|G|]$. A $\lambda$-walk in $\G^{\operatorname{mon}}$ connecting $g$ and $g'$ is simply given by a tuple $\bg$ in the unique relation of $\G^{\operatorname{mon}}$ satisfying $g_{\lambda_1}=g$ and $g_{\lambda_2}=g'$. Using the fact that $g$ and $g'$ have inverses in $G$, it is straightforward to check that such $\bg$ must exist. It follows that $\G^{\operatorname{mon}}$---and, thus, $\G$---is aperiodic with mixing time $\tau=1$.
Since, in addition, $\G$ is loopless when $|G|\geq 2$, Theorem~\ref{thm_aperiodic_quasi_linear_general_case} can be applied to $\CSP(\G)=\PCSP(\G,\G)$, and it yields the well-known fact that equations over groups have  linear width~\cite{Feder98:monotone}. We note that this argument also yields linear width for \emph{promise} equations over groups, as recently studied in~\cite{LZ24:icalp}.
\end{rem}

We conclude this section by proving the following monotonicity property of aperiodic structures.

\begin{prop}
\label{prop_monotonicity_aperiodicity}
Let $\A'\to\A$ be relational structures such that $\A'$ is aperiodic. Then, some induced substructure of $\A$ is aperiodic. 
\end{prop}
\begin{proof}
    Suppose first that $\A'$ and $\A$ are monic. Take a homomorphism $f:\A'\to\A$, and let $\tilde\A$ be the substructure of $\A$ induced by the range of $f$. Let $\tau$ be the mixing time of $\A'$, and choose two vertices $a,b\in \tilde \A$ and a $\tau$-pattern $\lambda$. Since $f$ is surjective onto $\tilde\A$, there exist two vertices $a',b'\in A'$ such that $f(a')=a$ and $f(b')=b$. Using that $\A'$ has mixing time $\tau$, we can find a $\lambda$-walk in $\A'$ consisting of vertices $a'=a_0,\dots,a_\tau=b'$ and tuples $\ba^{(1)},\dots,\ba^{(\tau)}$ in $R^{\A'}$. Then, one easily checks that the vertices $f(a_0),\dots,f(a_\tau)$ and the tuples $f(\ba^{(1)}),\dots,f(\ba^{(\tau)})$ in $R^{\tilde\A}$ yield a $\lambda$-walk in $\tilde\A$ connecting $a=f(a_0)$ to $b=f(a_\tau)$. This shows that $\tilde\A$ is aperiodic, as required.

    Suppose now that $\A'$ and $\A$ are not monic. Observe that $\A'^{\operatorname{mon}}\to\Amon$ and, by Definition~\ref{defn_aperiodicity_general}, $\A'^{\operatorname{mon}}$ is aperiodic. We deduce from the first part of the proof that there exists some nonempty subset $S\subseteq A$ such that the substructure of $\Amon$ induced by $S$---in symbols, $\Amon\big|_S$---is aperiodic. It is not hard to check that $\Amon\big|_S=(\A\big|_S)^{\operatorname{mon}}$. We deduce that $\A\big|_S$ is aperiodic, thus concluding the proof.
\end{proof}

\section{Fibrosity}
\label{sec_fibrosity}

Recall that, given an ($r$-uniform) hypergraph $\HH$, $\fbr{\tau}{\HH}$ denotes the maximum cardinality of a set of mutually disjoint $\tau$-fibers in $\HH$, while $\fbr{\max}{\HH}$, $\pi_\HH$, and $\lambda_\HH$ denote the number of  maximal fibers, pendent hyperedges, and links in $\HH$, respectively.
The goal of this section is to prove the following result.

\begin{thm*}
[Theorem~\ref{thm_sparse_implies_large_fibrosity_plus_pendency} restated]
    For $\tau\in\N$ and $1<\beta\in\R$, let $\HH$ be a hereditarily $\beta$-sparse hypergraph of girth at least $\tau$ having no isolated vertices. Then,
    \begin{align*}
        \fbr{\tau}{\HH}+\pi_\HH
        >
        \left(\frac{1}{10r\tau}-\beta+1\right)n_\HH.
    \end{align*}
\end{thm*}
Given a hyperedge $e$, we let $\sdr(e)$ be the \emph{sum of the degree reciprocals} of $e$; i.e., $\sdr(e)=\sum_{v\in e}\frac{1}{\deg_\HH(v)}$. One easily checks that the equality 
\begin{align}
\label{eqn_fundamental_eqn_of_sdr}
    n_\HH
    =
    \sum_{e\in\Eset(\HH)}\sdr(e)
\end{align}
holds, provided that $\HH$ has no isolated vertices.
We point out that, if $f$ and $f'$ are two disjoint $\tau$-fibers in $\HH$ consisting of the hyperedges $e_1,\dots,e_\tau$ and $e_1',\dots,e_\tau'$, respectively, it is not forbidden that, for some $i,j\in[\tau]$, the hyperedges $e_i$ and $e_j'$ are adjacent; however, they cannot be equal.

\begin{prop*}
[Proposition~\ref{prop_maximal_fibers_upper_bound} restated]
    For $\beta>1$, let $\HH$ be a hereditarily $\beta$-sparse hypergraph all of whose fibers are non-degenerate. Then,
    \begin{align*}
        \fbr{\max}{\HH}<3(\beta-1)n_\HH+3\pi_\HH.
    \end{align*}
\end{prop*}
\begin{proof}
The result is clear if $\HH$ consists in an isolated vertex or a single hyperedge, so we assume that this is not the case. Furthermore, since the quantities $\fbr{\max}{\HH}$, $n_\HH$, and $\pi_\HH$ are all additive with respect to disjoint unions, and since a subhypergraph of a hereditarily $\beta$-sparse hypergraph is hereditarily $\beta$-sparse itself, we can assume without loss of generality that $\HH$ is a connected hypergraph.

Let $\mathscr{W}$ be the set of pendent hyperedges of $\HH$ that are adjacent to some link. Moreover, let $\mathscr{Z}$ be the set containing every hyperedge of $\HH$ that is not pendent, is not a link, and is adjacent to some link.  It shall be useful to label the hyperedges in $\mathscr{Z}$ according to the number of links adjacent to them. Note that this number is at least $1$ (by definition of $\mathscr{Z}$) and at most $r$ (since, by definition of a link, any vertex $v\in e\in\mathscr{Z}$ is incident to at most one link).
Thus, for $i\in [r]$, we let $\mathscr{Z}_i$ be the set of hyperedges in $\mathscr{Z}$ that are adjacent to exactly $i$-many links.
We also let $w$, $z$, and $z_i$ denote the cardinalities of the sets $\mathscr{W}$, $\mathscr{Z}$, and $ \mathscr{Z}_i$, respectively.
The result is obtained by estimating the total $\sdr$ of $\mathscr{Z}$, and comparing it to the total $\sdr$ of all other hyperedges.

Take $i\in[r]$ and $e\in\mathscr{Z}_i$, and observe that $e$ intersects exactly $i$ links in exactly $i$ different vertices. Since the maximum degree of a vertex in a link is $2$, we deduce that $\sdr(e)\leq r-i+\frac{i}{2}=r-\frac{i}{2}$. However, we can improve this bound in the case that $i$ is $1$ or $2$. In the latter case, if $\sdr(e)$ were exactly $r-\frac{i}{2}=r-1$, it would follow that $e$ has exactly two vertices of degree $2$ and $r-2$ vertices of degree $1$; i.e., $e$ would need to be a link, which is impossible by the definition of $\mathscr{Z}$. As a consequence, at least one of the $r-2$ vertices of $e$ that are not incident to a link has degree at least $2$. This yields $\sdr(e)\leq r-3+1+\frac{1}{2}=r-\frac{i}{2}-\frac{1}{2}$. If $i=1$, one can do even better. Let $V$ be the set of $r-1$ vertices of $e$ that are not incident to a link. If all of them had degree $1$, $e$ would be a pendent hyperedge, against our assumption. Furthermore, if one of the vertices of $V$ had degree $2$ and all the others had degree $1$, $e$ would be a link, again contradicting the hypothesis. It follows that either two vertices of $V$ have degree at least $2$, or a vertex of $V$ has degree at least $3$. The latter choice yields the highest $\sdr$, and thus we can use it to give an upper bound: $\sdr(e)\leq r-2+\frac{1}{2}+\frac{1}{3}=r-\frac{i}{2}-\frac{2}{3}$. As a consequence, we bound the total $\sdr$ of $\mathscr{Z}$ as follows:  
\begin{align*}
\notag
    \sum_{e\in\mathscr{Z}}\sdr(e)
    &=
    \sum_{e\in\mathscr{Z}_1}\sdr(e)
    +
    \sum_{e\in\mathscr{Z}_2}\sdr(e)
    +
    \sum_{i\geq 3}\sum_{e\in\mathscr{Z}_i}\sdr(e)\\
    \notag
    &\leq 
    z_1\left(r-\frac{1}{2}-\frac{2}{3}\right)
    +
    z_2\left(r-1-\frac{1}{2}\right)
    +
    \sum_{i\geq 3}z_i\left(r-\frac{i}{2}\right)\\
    \notag
    &=
    r\sum_{i\in[r]}z_i-\frac{1}{2}\sum_{i\in[r]}iz_i-\frac{2}{3}z_1-\frac{1}{2}z_2\\
    &=
    (r-1)z+z-\frac{1}{2}\sum_{i\in[r]}iz_i-\frac{2}{3}z_1-\frac{1}{2}z_2,
\end{align*}
where we have used that $\sum_{i\in [r]}z_i=z$. Note that
\begin{align*}
    z-\frac{1}{2}\sum_{i\in[r]}iz_i-\frac{2}{3}z_1-\frac{1}{2}z_2
    &=
    \sum_{i\in [r]}\left(z_i-\frac{iz_i}{2}\right)-\frac{2}{3}z_1-\frac{1}{2}z_2
    =
    -\frac{1}{6}z_1-\frac{1}{2}z_2-\sum_{i\geq 3}\left(\frac{i}{2}-1\right)z_i\\
    &\leq
    -\frac{1}{6}\sum_{i\in[r]}iz_i.
\end{align*}
Therefore, we obtain
\begin{align}
\label{eqn_1221_1103}
    \sum_{e\in\mathscr{Z}}\sdr(e)
    \leq
    (r-1)z-\frac{1}{6}\sum_{i\in[r]}iz_i.
\end{align}

We say that a fiber $f$ \emph{connects} two 
(possibly equal) 
hyperedges $e,e'\in\Eset(\HH)\setminus f$ if some hyperedge in $f$ is adjacent to $e$ and some hyperedge in $f$ is adjacent to $e'$. 
Consider the multigraph (with loops allowed) $\mathcal{G}$ whose vertex set is $\mathscr{Z}\cup\mathscr{W}$ and whose edge multiset is defined as follows: For $e,e'\in\mathscr{Z}\cup\mathscr{W}$ (possibly equal), we insert one edge in $\mathcal{G}$ joining $e$ and $e'$ for each maximal fiber
of $\HH$ connecting $e$ and $e'$.
We claim that $m_{\mathcal{G}}=\fbr{\max}{\HH}$ (where $m_{\mathcal{G}}$ is the number of edges of the multigraph ${\mathcal{G}}$, counted with multiplicity). First, since any maximal fiber yields at most one edge in ${\mathcal{G}}$, it is clear that $m_{\mathcal{G}}\leq \fbr{\max}{\HH}$. 
Suppose that $m_{\mathcal{G}}<\fbr{\max}{\HH}$. This would mean that there exists a maximal fiber $f$ such that, for each $e\in f$, $e$ is only adjacent to other links. Since $f$ is maximal, $e$ is in fact only adjacent to other hyperedges in $f$. In other words, $f$ induces a cycle, which is impossible as $f$ is non-degenerate by assumption. It follows that the claim is true. 
For each $e\in\mathscr{Z}\cup\mathscr{W}$, let $\deg_{\mathcal{G}}(e)$ be 
the degree of $e$ in $\mathcal{G}$; i.e.,
the number of edges of $\mathcal{G}$ to which $e$ is incident (where a loop contributes $2$). 
Note that, if $e\in\mathscr{Z}_i$ for some $i\in [r]$, $\deg_{\mathcal{G}}(e)=i$. Moreover, if $e\in\mathscr{W}$, $\deg_{\mathcal{G}}(e)=1$.
We can now express the number of edges in $\mathcal{G}$ in terms of the sum of the vertex degrees via the handshaking lemma, thereby deducing that
\begin{align*}
    2\fbr{\max}{\HH}
    =
    2m_{\mathcal{G}}=
    \sum_{e\in\mathscr{Z}}\deg_{\mathcal{G}}(e)+\sum_{e\in\mathscr{W}}\deg_{\mathcal{G}}(e)
    =
    \sum_{i\in[r]}iz_i+w.
\end{align*}
Plugging this information into~\eqref{eqn_1221_1103} yields
\begin{align*}
    \sum_{e\in\mathscr{Z}}\sdr(e)
    \leq
    (r-1)z-\frac{1}{3}\fbr{\max}{\HH}+\frac{1}{6}w.
\end{align*}
Let $\mathscr{P}$ be the set of pendent hyperedges of $\HH$.
Since $\HH$ is connected and does not consist in a single hyperedge, $\sdr(e)\leq r-\frac{1}{2}$ for each $e\in\Eset(\HH)$. Moreover, $\sdr(e)\leq r-1$ if $e\not\in\mathscr{P}$. Noting also that, by assumption, $\HH$ does not contain isolated vertices,~\eqref{eqn_fundamental_eqn_of_sdr} yields
\begin{align*}
n_\HH
&=
\sum_{e\in\Eset(\HH)}\sdr(e)
=
\sum_{e\in\mathscr{P}}\sdr(e)
+
\sum_{e\in\mathscr{Z}}\sdr(e)
+
\sum_{e\in\Eset(\HH)\setminus (\mathscr{P}\cup\mathscr{Z})}\sdr(e)\\
&\leq
\left(r-\frac{1}{2}\right)\pi_\HH
+
(r-1)z-\frac{1}{3}\fbr{\max}{\HH}+\frac{1}{6}w
+
(m_\HH-\pi_\HH-z)(r-1)\\
&\leq
-\frac{1}{3}\fbr{\max}{\HH}+\pi_\HH+m_\HH(r-1),
\end{align*}
where we have used that $w\leq\pi_\HH$. We now make use of the fact that $\HH$ is hereditarily $\beta$-sparse and thus, in particular, it is $\beta$-sparse. This gives
\begin{align*}
    n_\HH
    <
    -\frac{1}{3}\fbr{\max}{\HH}+\pi_\HH+\beta n_\HH,
\end{align*}
which concludes the proof.
\end{proof}

We point out that, in the proof of Proposition~\ref{prop_maximal_fibers_upper_bound}, we are crucially using that the upper bound $\sdr(e)\leq r-\frac{i}{2}$ on the sum of the degree reciprocals of a hyperedge $e\in\mathscr{Z}_i$ can be improved in the case that $i=1$ or $i=2$. Without the improvement, the resulting upper bound on the number of maximal fibers of $\HH$ would be $\fbr{\max}{\HH}<(\beta-1)n_\HH+\pi_\HH+z$, and the presence of $z$ in the new expression would make the proposition useless towards proving Theorem~\ref{thm_sparse_implies_large_fibrosity_plus_pendency}. 

\begin{prop*}
[Proposition~\ref{prop_many_links_if_no_pendent_edges} restated]
    For $\beta>1$, let $\HH$ be a hereditarily $\beta$-sparse hypergraph having no isolated vertices. Then, 
    \begin{align*}
        \lambda_\HH>\left(\frac{1}{r}+6-6\beta\right)n_\HH-7\pi_\HH.
    \end{align*}
\end{prop*}

\begin{proof}
Like in the proof of Proposition~\ref{prop_maximal_fibers_upper_bound}, we can assume without loss of generality that $\HH$ is connected.
Let $\mathscr{P}$ and $\mathscr{L}$ denote the set of pendent hyperedges and the set of links of $\HH$, respectively. 
Similarly to Proposition~\ref{prop_maximal_fibers_upper_bound}, the result is obtained by comparing the total $\sdr$ of a certain set of hyperedges to that of all other hyperedges. In this case, instead of the set $\mathscr{Z}$, we consider the set $\mathscr{L}$. 

Notice that $\sdr(e)\leq r$ for any hyperedge $e\in\Eset(\HH)$ and $\sdr(e)=r-1$ if $e\in\mathscr{L}$. Furthermore, if $e\not\in\mathscr{P}\cup\mathscr{L}$, there are two possibilities: Either $(i)$ three or more vertices of $e$ have degree at least two, or $(ii)$ one vertex of $e$ has degree at least three, one other vertex has degree at least two, and all remaining vertices have degree one; in both cases, we find $\sdr(e)\leq r-\frac{7}{6}$. 
Since, by assumption, $\HH$ has no isolated vertices,~\eqref{eqn_fundamental_eqn_of_sdr} gives
\begin{align*}
    n_\HH
    &=
    \sum_{e\in\Eset(\HH)}\sdr(e)
    =
    \sum_{e\in\mathscr{P}}\sdr(e)
    +
    \sum_{e\in\mathscr{L}}\sdr(e)
    +\sum_{e\in\Eset(\HH)\setminus(\mathscr{P}\cup\mathscr{L})}\sdr(e)\\
    &\leq
    \pi_\HH r +\lambda_\HH (r-1)+(m_\HH-\pi_\HH-\lambda_\HH)\left(r-\frac{7}{6}\right)
    =
    \frac{1}{6}\lambda_\HH+\frac{7}{6}\pi_\HH+m_\HH(r-1)-\frac{1}{6}m_\HH.
\end{align*}
Note that $m_\HH$ appears twice, with opposite signs, in the right-hand side of the expression above. Hence, to continue the chain of inequalities, we need both an upper and a lower bound on $m_\HH$. The upper bound is the one dictated by the fact that $\HH$ is $\beta$-sparse; i.e., $m_\HH<\frac{\beta}{r-1}n_\HH$. The lower bound comes from the fact that $\HH$ is connected and, thus, $m_\HH\geq\frac{n_\HH-1}{r-1}$ (as can be easily shown by induction on $m_\HH$). Moreover, since $\HH$ contains no isolated vertices, $n_\HH\geq r$, and it follows that $\frac{n_\HH-1}{r-1}\geq\frac{n_\HH}{r}$. Therefore, we obtain
\begin{align*}
    n_\HH
    <
    \frac{1}{6}\lambda_\HH+\frac{7}{6}\pi_\HH+\beta n_\HH-\frac{1}{6r}n_\HH,
\end{align*}
which directly yields the claimed inequality.
\end{proof}

\begin{prop*}
[Proposition~\ref{prop_nice_property_fibrosity_max_fibrosity} restated]
    For any hypergraph $\HH$ and any $\tau\in\N$,
\begin{align*}
        \fbr{\tau}{\HH}+\fbr{\max}{\HH}
        >
        \frac{\lambda_\HH}{\tau}.
    \end{align*}
\end{prop*}
\begin{proof}
    For $i\in\N$, let $p_i$ be the number of maximal fibers of $\HH$ of size exactly $i$. Note that 
    \begin{align*}
    \sum_{i\in\N}p_i=\fbr{\max}{\HH}.
    \end{align*}
    Moreover, since the maximal fibers of $\HH$ partition the set of links of $\HH$, we have that 
    \begin{align*}
    \sum_{i\in\N}ip_i=\lambda_\HH.
    \end{align*}
    Observe now that a maximal fiber of size $i$ can be split into a set of $\lfloor\frac{i}{\tau}\rfloor$ mutually disjoint $\tau$-fibers (plus an additional $(i-\tau\lfloor\frac{i}{\tau}\rfloor)$-fiber). As a consequence, we find that
    \begin{align*}
        \fbr{\tau}{\HH}
        =
        \sum_{i\in\N}p_i\left\lfloor\frac{i}{\tau}\right\rfloor
        >
        \sum_{i\in\N}p_i\left(\frac{i}{\tau}-1\right)
        =
        \frac{\lambda_\HH}{\tau}-\fbr{\max}{\HH},
    \end{align*}
    and the conclusion follows.
\end{proof}
Theorem~\ref{thm_sparse_implies_large_fibrosity_plus_pendency} follows by combining Propositions~\ref{prop_maximal_fibers_upper_bound},~\ref{prop_many_links_if_no_pendent_edges}, and~\ref{prop_nice_property_fibrosity_max_fibrosity}, as detailed in Subsection~\ref{subsec_fibrosity}.

\section{Sparsity}
\label{sec_sparsity}
In this section, we prove the existence of sparse and highly chromatic hypergraphs of large girth, as stated next.

\begin{thm*}
[Theorem~\ref{thm_sparse_incomparability_girth} restated]
Take two positive integer numbers $g$ and $h$ and a real number $\beta>1$. There exists a positive real number $\delta=\delta(g,h,\beta)$ and a positive integer number $n_0=n_0(g,h,\beta)$ such that, for each $n\geq n_0$, there exists a hypergraph $\HH$ with the following properties:
\begin{enumerate}
\itemsep-.2em 
    \item $\HH$ has $n$ vertices;
    \item $\girth(\HH)\geq g$;
    \item $\chr(\HH)\geq h$;
    \item $\HH$ is $(\delta n,\beta)$-threshold-sparse.
\end{enumerate}
\end{thm*}

As discussed in the Overview, our proof of Theorem~\ref{thm_sparse_incomparability_girth} is probabilistic, and it combines Erd\H{o}s--Hajnal's proof of the existence of hypergraphs with arbitrarily large girth and chromatic number~\cite{erdHos1966chromatic} with Atserias--Dalmau's result on the sparsity of random Erd\H{o}s-R\'enyi hypergraphs~\cite{Atserias22:soda}. The latter result is stated in terms of a different version of threshold-sparsity, which we now define.\footnote{We remark that the notion of threshold-sparsity given in Subsection~\ref{subsec_sparsity} is also implicit in~\cite{Atserias22:soda}.}

For $\mu$ and $\nu$ positive real numbers, we say that a hypergraph $\HH$ is \emph{$(\mu,\nu)$-vertex-threshold-sparse} if any subhypergraph $\HH'$ of $\HH$ such that $n_{\HH'}\leq\mu\cdot n_{\HH}$ satisfies $m_{\HH'}<\nu\cdot n_{\HH'}$.
The following simple result connects the two notions of threshold-sparsity.

\begin{lem}
\label{lem_vertex_sparse_edge_sparse}
    Let $\beta>1,\gamma,\mu,\nu$ be positive real numbers, let $\HH$ be a $(\mu,\nu)$-vertex-threshold-sparse hypergraph, and suppose that $\frac{\beta}{\nu}\geq r-1$ and $n_\HH\geq (r-1)\frac{\gamma}{\beta\mu}$. Then, $\HH$ is $(\gamma,\beta)$-threshold-sparse.
\end{lem}
\begin{proof}
    Let $\HH'$ be a subhypergraph of $\HH$ satisfying $m_{\HH'}\leq\gamma$, and suppose that $\HH'$ is not $\beta$-sparse. We obtain
    \begin{align*}
        n_{\HH'}
        \leq(r-1)\frac{m_{\HH'}}{\beta}
        \leq
        (r-1)\frac{\gamma}{\beta}
        \leq\mu\cdot n_\HH.
    \end{align*}
    Since $\HH$ is $(\mu,\nu)$-vertex-threshold-sparse, we deduce that \begin{align*}
        m_{\HH'}
        <
        \nu\cdot n_{\HH'}
        \leq
        \frac{\beta}{r-1}n_{\HH'},
    \end{align*}
    which contradicts the fact that 
    $\HH'$ is not $\beta$-sparse.
\end{proof}

For $n\in\N$ and $p\in\R$, $0\leq p\leq 1$, the expression $\HH\sim\mathscr{H}(n,p)$ shall denote that $\HH$ is a random ($r$-uniform) Erd\H{o}s-R\'enyi hypergraph with $n$ vertices, where each $r$-element set of vertices forms a hyperedge with probability $p$, independently.
The next result gives an upper bound on the probability that a random Erd\H{o}s-R\'enyi hypergraph is not vertex-threshold-sparse.
\begin{prop}[\cite{Atserias22:soda}]
\label{prop_AD_probability_alpha_beta_sparse}
Let $\mu,\nu,\ell$ be positive real numbers and let $n$ be a positive integer number. If $1\leq\ell\leq n^{r-1}$ and $\mu\leq(\frac{\nu}{\ell})^{\frac{1}{r-1}}(\frac{r}{e})^{\frac{r}{r-1}}$, the probability that a random hypergraph $\HH\sim\mathscr{H}(n,\ell n^{1-r})$ is not $(\mu,\nu)$-vertex-threshold-sparse is at most
\begin{align}
\label{eqn_1703_1149}
    \pi(r,\ell,n,\mu,\nu)
    =
    \sum_{i=1}^{\lfloor\mu n\rfloor}
    \left(\left(\frac{n}{i}\right)^{1-(r-1)\nu}\ell^\nu e^{1+(r+1)\nu}r^{-r\nu}\nu^{-\nu}\right)^i.
\end{align}
\end{prop}

\begin{proof}[Proof of Theorem~\ref{thm_sparse_incomparability_girth}]
We start by defining five parameters $\ell$, $\nu$, $\vartheta$, $\mu$, and $\delta$.
Intuitively, $\ell$ determines the hyperedge probability of the Erd\H{o}s-R\'enyi hypergraph we are soon going to sample; $\nu$ and $\mu$ express the desired vertex-threshold-sparsity of the hypergraph; $\delta$ controls the threshold-sparsity of a subhypergraph we shall obtain by removing all short cycles; and $\vartheta$ simply connects the other parameters in a way that all necessary inequalities in the proof hold.
\begin{align}
        \begin{array}{lll}
        \ell=\frac{(3hr)^r}{2h}\log(3eh)+1,\\[5pt]
         \nu=\frac{\beta}{r-1},\\[5pt]
        \vartheta=\ell^\nu e^{1+(r+1)\nu}r^{-r\nu}\nu^{-\nu},\\[5pt]
        \mu=\min
        \left(
        (\frac{\nu}{\ell})^{\frac{1}{r-1}}(\frac{r}{e})^{\frac{r}{r-1}},
        (3\vartheta)^{-\frac{1}{\beta-1}}
        \right),\\[5pt]
        \delta=\frac{\beta\mu}{r-1}.
    \end{array}
    \end{align}
    Notice that all of the parameters are positive real numbers. 
    Let $n$ be an even positive integer number to be determined later, and suppose that $n\geq\ell^{\frac{1}{r-1}}$. Let $p=\ell n^{1-r}$, and sample a hypergraph $\HH\sim\mathscr{H}(n,p)$.

    Fix $w=\lfloor\frac{n}{2h}\rfloor$, and let $W$ be the random variable counting the number of independent sets in $\HH$ of size $w$ (where an independent set is a set $S$ of vertices such that $e\not\subseteq S$ for any $e\in\Eset(\HH)$). 
    Suppose that $n\geq 2hr$ (which implies that $w\geq r$) and that $n\geq 6h$ (which implies that $w\geq\frac{n}{2h}-1\geq\frac{n}{3h}$).
    It is well known that, for $1\leq x\leq y$, the inequalities 
    \begin{align*}
        \left(\frac{y}{x}\right)^x\leq{y\choose x}\leq\left(\frac{ey}{x}\right)^x 
    \end{align*}
    hold (for example, see~\cite{cormen2022introduction}).
    As a consequence, using the linearity of expectation, we obtain the following upper bound for the expected value of $W$:
    \begin{align*}
        \mathbb{E}[W]
        =
        {n\choose w}(1-p)^{w\choose r}
        \leq
         {n\choose w}e^{-p{w\choose r}}
        \leq
        \left(\frac{en}{w}\right)^w e^{-p\left(\frac{w}{r}\right)^r}
        \leq
        (3eh)^{\frac{n}{2h}}e^{-\ell n^{1-r}\frac{n^r}{(3hr)^r}}
        =\left(
        \frac{(3eh)^\frac{1}{2h}}{e^{\frac{\ell}{(3hr)^r}}}
        \right)^n. 
    \end{align*}
    The choice of $\ell$ guarantees that 
   $
        (3eh)^\frac{1}{2h}
        <
        e^{\frac{\ell}{(3hr)^r}}$, whence it follows that
    $\mathbb{E}[W]$ approaches $0$ as $n$ tends to infinity. Recall that the independence number of $\HH$ (in symbols, $\ind(\HH)$) is the maximum cardinality of an independent set in $\HH$.
    By Markov's inequality, the probability that $W$ is at least $1$---i.e., that $\ind(\HH)\geq w$---is at most $\mathbb{E}[W]$ and, thus, it tends to $0$ as $n$ tends to infinity.
    In particular, for $n$ large enough, 
    \begin{align}
    \label{eqn_union_bound_1}
        \mathbb{P}\left(W\geq 1\right)<\frac{1}{4}.
    \end{align}

It is well known~\cite{berge1989hypergraphs} that a cycle of length $j\geq 2$ in a hypergraph yields a collection of distinct hyperedges $e_1,\dots,e_j$ such that $|\cup_{i\in[j]}e_i|\leq j(r-1)$ and, vice versa, a collection of $j$ distinct hyperedges satisfying this condition contains a cycle of length at most $j$. The number of such collections is at most
    \begin{align*}
        {n\choose j(r-1)}{j(r-1)\choose r}^j,
    \end{align*}
    which is upper-bounded by $x_j n^{j(r-1)}$ for some constant $x_j$ independent of $n$. 
    Let a \emph{short cycle} indicate a cycle of length at most $g-1$, and let $V$ be the random variable counting the number of short cycles in $\HH$. 
    By linearity of expectation, we find
    \begin{align*}
        \mathbb{E}[V]
        \leq
        \sum_{j=2}^{g-1}
        x_j n^{j(r-1)}p^j
        =
        \sum_{j=2}^{g-1}
        \ell^j x_j.
    \end{align*}
    Hence, using Markov's inequality, we deduce that the probability that $V$ is at least $\frac{n}{2}$ is at most
    \begin{align*}
        \frac{2\mathbb{E}[V]}{n}
        \leq \frac{2\sum_{j=2}^{g-1}
        \ell^j x_j}{n}.
    \end{align*}
    Since the numerator in the expression above is independent of $n$, the fraction approaches $0$ as $n$ tends to infinity. In particular, 
    \begin{align}
    \label{eqn_union_bound_2}
        \mathbb{P}\left(V\geq\frac{n}{2}\right)<\frac{1}{4}
    \end{align}
    for $n$ large enough.

Let $A$ be the event that $\HH$ is not $(\mu,\nu)$-vertex-threshold-sparse.
The choice of $\mu$ guarantees that 
Proposition~\ref{prop_AD_probability_alpha_beta_sparse} applies. Hence, 
the probability of $A$ is upper bounded by the quantity $\pi=\pi(r,\ell,n,\mu,\nu)$ in~\eqref{eqn_1703_1149}. From the fact that $\mu\leq (3\vartheta)^{-\frac{1}{\beta-1}}$, we deduce that $\mu^{\beta-1}\vartheta\leq\frac{1}{3}$. Therefore,
\begin{align}
\label{eqn_union_bound_3}
\mathbb{P}(A)
\leq
\pi=
    \sum_{i=1}^{\lfloor\mu n\rfloor}
    \left(\left(\frac{i}{n}\right)^{\beta-1}\vartheta\right)^i
    \leq
    \sum_{i=1}^{\lfloor\mu n\rfloor}
    \left(\mu^{\beta-1}\vartheta\right)^i
    \leq
    \sum_{i=1}^{\lfloor\mu n\rfloor}
    \frac{1}{3^i}
    <
    \sum_{i=1}^{\infty}
    \frac{1}{3^i}
    =
    \frac{1}{2}.
\end{align}

By the union bound, we deduce from~\eqref{eqn_union_bound_1},~\eqref{eqn_union_bound_2},~and~\eqref{eqn_union_bound_3} that, if $n$ is large enough, there exists a hypergraph $\HH$ having $n$ vertices such that $W=0$, $V<\frac{n}{2}$, and $A$ does not hold. In other words, $\ind(\HH)<w$, $\HH$ has fewer than $\frac{n}{2}$ short cycles, and $\HH$ is $(\mu,\nu)$-vertex-threshold-sparse. 

Select a set $S$ of vertices of $\HH$ of cardinality $\frac{n}{2}$ in a way that each short cycle of $\HH$ contains a vertex of $S$; in order for its cardinality to be exactly $\frac{n}{2}$, $S$ is allowed to contain vertices that do not appear in any short cycle.
Let $\tilde\HH$ be the subhypergraph of $\HH$ induced by $\Vset(\HH)\setminus S$, so that $n_{\tilde\HH}=\frac{n}{2}$. We have $\ind(\tilde\HH)\leq \ind(\HH)<w$, while $\girth(\tilde\HH)\geq g$ as all short cycles have been broken. Since the product of the chromatic and the independence numbers of a hypergraph is at least the number of its vertices (see~\cite{berge1989hypergraphs}), we find
\begin{align*}
    \chr(\tilde\HH)
    \geq
    \frac{n_{\tilde\HH}}{\ind(\tilde\HH)}
    \geq
    \frac{n}{2w}
    \geq h.
\end{align*}
Finally, observe that $\tilde{\HH}$ is $(\mu,\nu)$-vertex-threshold-sparse by monotonicity of vertex-threshold-sparsity. The choice of $\delta$ guarantees that
\begin{align*}
    (r-1)\frac{\delta n_{\tilde\HH}}{\beta\mu}
    =
    n_{\tilde\HH}.
\end{align*}
Therefore,  Lemma~\ref{lem_vertex_sparse_edge_sparse} implies that $\tilde\HH$ is $(\delta n_{\tilde\HH},\beta)$-threshold-sparse, thus concluding the proof of Theorem~\ref{thm_sparse_incomparability_girth}.
\end{proof}

\section{Consistency}
\label{sec_consistency}
The goal of this section is to establish that a $(\kappa,\gamma)$-consistency gap---for some suitable parameter $\gamma$---can be used to recover a ${\kappa}$-strategy and, thus, to certify acceptance by the ${\kappa}$-th level of the local-consistency algorithm applied to an aperiodic template and a sparse instance having large girth.  
\begin{thm*}
[Theorem~\ref{thm_consistent_means_bounded_width} restated]
    Let $\X$ and $\A$ be monic structures, and take two numbers ${\kappa}\in\N$ and $\gamma\in\R$ such that $\gamma\geq n_\A$. Suppose that 
    \begin{itemize}
    \itemsep-.2em 
        \item $\A$ is aperiodic with mixing time $\tau$;
        \item $\X$ is oriented, and $\Xsym$ has girth at least $\tau$ and is $(\gamma,\beta)$-threshold-sparse for some real number $1<\beta<1+\frac{1}{10r\tau}$;
        \item $(\X,\A)$ has a $({\kappa},\gamma)$-consistency gap.
    \end{itemize}
    Then, $\Lcon^{\kappa}(\X,\A)=\YES$.
\end{thm*}

The key to proving Theorem~\ref{thm_consistent_means_bounded_width} is to show that $\tau$-fibers and pendent hyperedges yield \emph{boundary sets} (as considered in~\cite{Atserias22:soda,molloy2007resolution}) for aperiodic templates having mixing time $\tau$, in the sense that any partial homomorphism to such templates can be extended to include $\tau$-fibers and pendent hyperedges. This is done in the next two propositions.

\begin{prop*}
[Proposition~\ref{prop_extension_tau_fibers} restated]
    Let $\X,\A$ be monic structures such that $\X$ is oriented and $\A$ is aperiodic with mixing time $\tau$, let $f$ be a $\tau$-fiber of $\Xsym$, let $J$ be a joint\footnote{Recall the definition of joints given in Subsection~\ref{subsec_consistency}.} of $f$, and let $\Y$ be the substructure of $\X$ induced by $(X\setminus \bigcup_{e\in f}e)\cup J$. Then, any homomorphism $h:\Y\to\A$ can be extended to a homomorphism $h':\X\to\A$.
\end{prop*}
\begin{proof}
Let $J=\{u,v\}$, where $u$ and $v$ are possibly equal vertices of $\HH_f$.
Label the links in $f$ as $e_1,\dots,e_\tau$, where $e_{i}$ is adjacent to $e_{i+1}$ for each $i\in[\tau-1]$. Assume, without loss of generality, that $u\in e_1$ and $v\in e_\tau$.
Set $w_0=u$ and define iteratively, for each $i\in[\tau]$, the vertex $w_i$ as the only vertex of $e_i\setminus\{w_{i-1}\}$ having degree $2$. Note that this implies 
$w_\tau=v$. 
Recall that $R$ denotes the unique relation symbol in the signature of $\X$ and $\A$, and $r=\ar(R)$. 
Since $\X$ is oriented, by the definition of $\Xsym$, there exist tuples $\bx^{(1)},\dots,\bx^{(\tau)}\in R^\X$ such that $e_i=\set(\bx^{(i)})$ for each $i\in[\tau]$, and the $\bx^{(i)}$'s are uniquely determined. Hence, there exist two tuples $\bp,\bq\in [r]^\tau$ such that, for each $i\in[\tau]$, $x^{(i)}_{p_i}=w_{i-1}$ and $x^{(i)}_{q_i}=w_{i}$. 
From the way we chose the vertices $w_i$, it follows that $w_{i-1}\neq w_i$ for each $i\in[\tau]$.
Therefore, $p_i\neq q_i$ for each $i\in[\tau]$. Hence, the matrix $\lambda=\begin{bmatrix}
        \bp&\bq
    \end{bmatrix}$ is a $\tau$-pattern. Observe that $w_0=u$ and $w_\tau=v$ are both vertices of $\Y$.
    Using the assumption that $\A$ has mixing time $\tau$, we can then find a $\lambda$-walk connecting $h(w_0)$ to $h(w_\tau)$. Suppose that the $\lambda$-walk consists of vertices $a_0=h(w_0),a_1,\dots,a_\tau=h(w_\tau)$ in $A$ and tuples $\ba^{(1)},\dots,\ba^{(\tau)}$ in $R^\A$. Consider the map $h':X\to A$ defined as follows: 
    \begin{itemize}
    \itemsep-.2em 
        \item[$(I)$] $h'(x^{(i)}_j)=a^{(i)}_j$ 
        for each $i\in [\tau]$ and each $j\in[r]$;
        \item[$(II)$] 
        $h'(x)=h(x)$ for each $x\in Y$.
    \end{itemize}
    First, we show that $h'$ is well defined. The two cases above cover all vertices of $X$, and the only vertices covered multiple times are the vertices $w_i$ for $i\in[\tau]$. Indeed, $w_0=x^{(1)}_{p_1}$, so we need to check that $h(w_0)=a^{(1)}_{p_1}$. This is true since, using the definition of $\lambda$-walk, we have
    \begin{align*}
        a^{(1)}_{p_1}
        =
        a^{(1)}_{\lambda_{1,1}}
        =
        a_0
        =
        h(w_0).
    \end{align*}
    Similarly, the fact that $w_\tau=x^{(\tau)}_{q_\tau}$ does not yield problems as 
    \begin{align*}
        a^{(\tau)}_{q_\tau}
        =
        a^{(\tau)}_{\lambda_{\tau,2}}
        =
        a_\tau
        =
        h(w_\tau).
    \end{align*}
    Furthermore, for $i\in[\tau-1]$, we have $w_i=x^{(i+1)}_{p_{i+1}}=x^{(i)}_{q_i}$. Also this identity is not problematic, as
    \begin{align*}
        a^{(i+1)}_{p_{i+1}}
        =
        a^{(i+1)}_{\lambda_{i+1,1}}
        =
        a_i
        =
        a^{(i)}_{\lambda_{i,2}}
        =
        a^{(i)}_{q_i}.
    \end{align*}
    Therefore, $h'$ is well defined.
    We now claim that it is a homomorphism from $\X$ to $\A$. Since $h'$ clearly extends $h$, this would be enough to conclude the proof. Take $\bx\in R^\X$; we need to show that $h'(\bx)\in R^\A$. If $\bx\in R^\Y$, the claim is clear as, in this case, $h'(\bx)=h(\bx)\in R^\A$ since $h$ is a homomorphism. Otherwise, by definition of fiber, it must be the case that $\set(\bx)\in f$. Thus, since $\X$ is oriented, we deduce that $\bx=\bx^{(i)}$ for some $i\in[\tau]$. Using the rule $(I)$ above, we conclude that $h'(\bx)=h'(\bx^{(i)})=\ba^{(i)}\in R^\A$, as required.
\end{proof}

\begin{prop*}
[Proposition~\ref{prop_extension_pendent_edges} restated]
    Let $\X,\A$ be monic structures such that $\X$ is oriented and $\A$ is aperiodic, let $e$ be a pendent hyperedge for $\Xsym$, let $J$ be a joint of $e$, and let $\Y$ be the substructure of $\X$ induced by $(X\setminus e)\cup J$. Then, any homomorphism $h:\Y\to\A$ can be extended to a homomorphism $h':\X\to\A$.
\end{prop*}
\begin{proof}
Call $v$ the unique vertex in $J$,
let $\bx\in R^\X$ be such that $e=\set(\bx)$, and let $j\in[r]$ be such that $x_j=v$, where $r$ is the arity of $R$. Let $\tau$ be the mixing time of $\A$. The definition of oriented structure guarantees that $r\geq 2$; hence, there exists some $\tau$-pattern $\lambda$ such that $\lambda_{1,1}=j$. Using the aperiodicity of $\A$, we find a $\lambda$-walk connecting $h(v)$ to itself. In particular, this means that there exists a tuple $\ba\in R^\A$ such that $h(v)=a_{\lambda_{1,1}}=a_j$. Since all vertices in $e\setminus\{v\}$ have degree $1$ in $\Xsym$, the map $h':X\to A$ given by $h'(x_i)=a_i$ for each $i\in [r]$ and $h'(x)=h(x)$ for all other vertices of $\X$ is easily seen to yield a homomorphism from $\X$ to $\A$ extending $h$.    
\end{proof}

We can now prove Theorem~\ref{thm_consistent_means_bounded_width}. This result extends~\cite[Lemmas~8,10]{Atserias22:soda}, and the proof below is an adaptation of the one in~\cite{Atserias22:soda}. Given a monic structure $\Y$ whose unique relation symbol is $R$, we let $m_\Y=|R^\Y|$.
\begin{proof}[Proof of Theorem~\ref{thm_consistent_means_bounded_width}]
If $r=2$ and $\tau=1$, it follows from Definition~\ref{defn_aperiodicity_monic} that $\A$ has a loop, which implies that $\X\to\A$ and, thus, that $\Lcon^\kappa(\X,\A)=\YES$, as needed; hence, we assume that $r\geq 3$ or $\tau\geq 2$ (note that the case $r=1$ is forbidden by the definition of oriented structure).
Similarly, we can assume that $m_\X\neq 0$.
Let $\hat\X$ be the substructure of $\X$ obtained by removing all vertices that do not belong to any tuple in $R^\X$. It is easy to check from Definition~\ref{defn_consistency} that a $(\kappa,\gamma)$-consistency gap for $(\X,\A)$ implies a $(\kappa,\gamma)$-consistency gap for $(\hat\X,\A)$. Moreover, any $\kappa$-strategy $\hat{\mathscr{F}}$ for $\hat\X$ and $\A$ can be straightforwardly turned into a $\kappa$-strategy for $\X$ and $\A$ by extending the partial homomorphisms in $\hat{\mathscr{F}}$ in an arbitrary way to cover the vertices of $X\setminus \hat X$. This means that we can assume, without loss of generality, that $\Xsym$ has no isolated vertices. Let $\mathscr{F}$ be the set of all functions from a subset of $X$ of size at most ${\kappa}$ to $A$, that are consistent with every substructure of $\X$ whose relation has at most $\gamma$ tuples. We claim that $\mathscr{F}$ is a ${\kappa}$-strategy for $\X$ and $\A$.

    Take a nonempty map $f\in\mathscr{F}$ having domain $X'\subseteq X$, let $\X'$ be the substructure of $\X$ induced by $X'$, let $\bx\in R^{\X'}$, and let $\Y$ be the substructure of $\X$ whose vertex set is $Y=\set(\bx)$ and whose relation is $R^\Y=\{\bx\}$. Since $m_\Y=1\leq n_\A\leq\gamma$, $f$ is consistent with $\Y$ by definition of $\mathscr{F}$, so there exists a homomorphism $h:\Y\to\A$ that agrees with $f$ on $X'\cap Y=Y$. It follows that $f(\bx)=h(\bx)\in R^\A$. Therefore, $f$ is a homomorphism from $\X'$ to $\A$ (and, thus, a partial homomorphism from $\X$ to $\A$). 

    We now claim that all substructures $\Y$ of $\X$ such that $m_\Y\leq\gamma$ are homomorphic to $\A$. Otherwise, let $\Y$ be a witness of the contrary; we may assume without loss of generality that any proper substructure of $\Y$ is homomorphic to $\A$ and that $m_\Y\geq 1$. Note that $\Y$ is oriented and $\Ysym$ is
    hereditarily $\beta$-sparse.
    Applying Theorem~\ref{thm_sparse_implies_large_fibrosity_plus_pendency}, we deduce that 
    \begin{align*}
        \fbr{\tau}{\Ysym}+\pi_\Ysym
        >
        \left(\frac{1}{10r\tau}-\beta+1\right)n_\Ysym
        >
        0,
    \end{align*}
    where we have used that $\beta<1+\frac{1}{10r\tau}$.
    Hence, $\Ysym$ has at least one $\tau$-fiber or 
    pendent hyperedge. In the former case, let $S=\bigcup_{e\in f}e$ where $f$ is a $\tau$-fiber of $\Ysym$, and let $J$ be a joint of $f$; in the latter case, let $S=e$ where $e$ is a pendent hyperedge of $\Ysym$, and let $J$ be a joint of $e$.
    Let $\Y'$ be the substructure of $\Y$ induced by $(Y\setminus S)\cup J$. 
    Since at least one of the conditions $r\geq 3$ and $\tau\geq 2$ holds---as assumed at the beginning of the proof---we observe that $S\setminus J\neq\emptyset$. Thus, $n_{\Y'}<n_\Y$.
    By the minimality of $\Y$, we have that $\Y'\to\A$.     
    It follows from Proposition~\ref{prop_extension_tau_fibers} or from Proposition~\ref{prop_extension_pendent_edges} that $\Y\to\A$, a contradiction.
    Therefore, the claim is true. In particular, we deduce that the empty function from $\emptyset$ to $A$ belongs to $\mathscr{F}$, thus yielding $\mathscr{F}\neq\emptyset$.

    The fact that $\mathscr{F}$ is closed under restrictions directly follows from the definition of consistent functions. 
To show that $\mathscr{F}$ has the extension property up to ${\kappa}$, take $\mathscr{F}\ni f:X'\to A$ with $|X'|<{\kappa}$ and pick $x\in X\setminus X'$. For $a\in A$, consider the map $f_a:X'\cup\{x\}\to A$ mapping $x$ to $a$ and $x'$ to $f(x')$ for each $x'\in X'$. Suppose that for each $a\in A$ there exists some substructure $\X_a$ of $\X$ whose relation has at most $\frac{\gamma}{n_\A}$ tuples such that $f_a$ is not consistent with $\X_a$, and let $\textbf{Z}=\bigcup_{a\in A}\X_a$ (i.e., $Z=\bigcup_{a\in A}X_a$ and $R^{\textbf{Z}}=\bigcup_{a\in A}R^{\X_a}$). Note that, for each $a\in A$, $f$ is consistent with $\X_a$ since $f\in\mathscr{F}$; hence, $x\in\bigcap_{a\in A}X_a\subseteq Z$. 
If $f$ is consistent with $\textbf{Z}$, there exists a homomorphism $h:\textbf{Z}\to\A$ such that $f$ and $h$ agree on $X'\cap Z$. Using that $x\in Z$, observe that the restriction of $h$ to $\X_{h(x)}$ is a homomorphism from $\X_{h(x)}$ to $\A$ that agrees with $f_{h(x)}$ on $(X'\cup\{x\})\cap X_{h(x)}$. This means that $f_{h(x)}$ is consistent with $\X_{h(x)}$, a contradiction.
As a consequence,  $f$ is not consistent with $\textbf{Z}$. But this contradicts the fact that $f\in\mathscr{F}$, as
\begin{align*}
m_{\textbf{Z}}\leq\sum_{a\in A}m_{\X_a}\leq\sum_{a\in A}\frac{\gamma}{n_\A}=\gamma.
\end{align*}
It follows that there exists some $a'\in A$ such that $f_{a'}$ is consistent with every substructure of $\X$ whose relation has at most $\frac{\gamma}{n_\A}$ tuples. Let $\W$ be the substructure of $\X$ induced by $X'\cup\{x\}$.
Since, by assumption, $\frac{\gamma}{n_\A}\geq 1$, the same argument as in the beginning of the proof shows that $f_{a'}$ is a homomorphism from $\W$ to $\A$. Moreover, $n_\W\leq {\kappa}$. Using the fact that $(\X,\A)$ has a $({\kappa},\gamma)$-consistency gap, we deduce that $f_{a'}$ is consistent with every substructure of $\X$ whose relation has at most $\gamma$ tuples; i.e., $f_{a'}\in\mathscr{F}$. Since, clearly, $f_{a'}$ extends $f$, this means that $\mathscr{F}$ has the extension property up to $\kappa$, as required.

In conclusion, we have shown that $\mathscr{F}$ is a ${\kappa}$-strategy for $(\X,\A)$ and, thus, $\Lcon^{\kappa}(\X,\A)=\YES$.
\end{proof}

\section{Everything together}
In this section, we put together all the pieces of the puzzle and prove our main result.

\begin{thm*}[Theorem~\ref{thm_aperiodic_quasi_linear_general_case} restated]
    Let $\A\to\B$ be relational structures such that $\A$ is aperiodic and $\B$ is loopless. Then, $\PCSP(\A,\B)$ has linear width.
\end{thm*}
\noindent

It shall be convenient to only deal with monic structures.
The next result shows that this assumption does not yield a loss of generality. While it is essentially folklore, we give a self-contained proof in the Appendix for completeness.

\begin{lem}
    \label{lem_enough_to_consider_monic}
    Let $\A\to\B$ be relational structures, and suppose that $\PCSP(\Amon,\Bmon)$ has linear width. Then, $\PCSP(\A,\B)$ has linear width, too. 
\end{lem}

Let $\Y$ be a monic structure whose relation has arity $r$.
For $c\in\N$, let $\K_{r,c}$ be the monic structure with domain $[c]$ whose unique, $r$-ary relation is the subset of $[c]^r$ obtained by removing all constant tuples. If $\Y$ is loopless, its \emph{chromatic number} is the minimum $c$ such that $\Y\to\K_{r,c}$. 

\begin{proof}[Proof of Theorem~\ref{thm_aperiodic_quasi_linear_general_case}]
Observe that a structure $\Y$ is aperiodic (resp. loopless) if and only if $\Y^{\operatorname{mon}}$ is aperiodic (resp. loopless). Hence, by virtue of Lemma~\ref{lem_enough_to_consider_monic},
it is enough to prove the result in the case that $\A$ and $\B$ are monic structures.
Let $\tau$ be the mixing time of $\A$ and let $c$ be the chromatic number of $\B$; since $\A$ is aperiodic and $\B$ is loopless, these numbers are both finite.
Also, let $R$ be the unique symbol in their signature and
let $r$ be the arity of $R$. If $r=1$ or if $r=2$ and $\tau=1$, Definition~\ref{defn_aperiodicity_monic} (and Footnote~\ref{cumbersome_footnote}) would imply that $\A$ contains a loop, which is impossible since $\A\to\B$ and $\B$ is loopless. Thus, we can assume that either $r\geq 3$ or $\tau\geq 2$ (as in the proof of Theorem~\ref{thm_consistent_means_bounded_width}).

Pick a real number $\beta$ such that $1<\beta<1+\frac{1}{10 r\tau}$, and apply Theorem~\ref{thm_sparse_incomparability_girth} to the parameters $g=\tau$, $h=c+1$, and $\beta$. Let $\delta=\delta(\tau,c+1,\beta)$ and $n_0=n_0(\tau,c+1,\beta)$ be as in the statement of Theorem~\ref{thm_sparse_incomparability_girth}, and set
\begin{align}
\label{eqn_1618_1703}
    \epsilon
    =
    \min\left(
    \frac{1}{n_0},
    \left(\frac{1}{10r\tau}-\beta+1\right)\frac{(r-1)\delta}{\beta n_\A}
    \right).
\end{align}
Observe that $\epsilon>0$. We claim that $\epsilon$ is a constant witnessing that $\PCSP(\A,\B)$ has linear width (see Section~\ref{sec_prelimns}).

Take $n\in\N$. We can assume without loss of generality that $n\geq\frac{1}{\epsilon}$. Otherwise, we might simply take as $\X$ any instance on $n$ vertices having a loop. Since $\B$ is loopless, $\X\not\to\B$. Moreover, $\Lcon^{\lfloor\epsilon n\rfloor}(\X,\A)=\Lcon^{0}(\X,\A)=\YES$, as needed. 
The fact that $n\geq\frac{1}{\epsilon}$ implies that $n\geq n_0$, via~\eqref{eqn_1618_1703}. Hence, invoking Theorem~\ref{thm_sparse_incomparability_girth}, we find a $(\delta n,\beta)$-threshold-sparse $r$-uniform hypergraph $\HH$ such that $n_\HH=n$, $\girth(\HH)\geq\tau$, and $\chr(\HH)\geq c+1$.
Consider now a structure $\X$ obtained from $\HH$ by choosing an arbitrary orientation for each hyperedge; i.e., $\X$ has domain $\Vset(\HH)$ and its unique relation $R^\X$ (of arity $r$) contains, for each $e\in\Eset(\HH)$, precisely one $r$-tuple $\bx$ such that $\set(\bx)=e$. Note that $\X$ is an oriented monic structure, and $\Xsym=\HH$. We claim that the pair $(\X,\A)$ has a $(\lfloor\epsilon n\rfloor,\delta n)$-consistency gap.  
Take a homomorphism $f:\W\to\A$, where $\W$ is a substructure of $\X$ with $n_\W\leq \lfloor\epsilon n\rfloor$. 
Suppose, for the sake of contradiction, that there exists a substructure $\Y$ of $\X$ such that $f$ is not consistent with $\Y$ and $\frac{\delta n}{n_\A}< m_\Y\leq\delta n$. Assume, without loss of generality, that $\Y$ is minimal, in the sense that $f$ is consistent with any proper substructure of $\Y$.
Let $\GG=\Ysym$, and let $\GG'$
be the subhypergraph of $\GG$ obtained by removing all isolated vertices of $\GG$.
By the monotonicity of threshold-sparsity, we deduce that $\GG'$ is hereditarily $\beta$-sparse.
Moreover, its girth is at least $\tau$. 
It follows from Theorem~\ref{thm_sparse_implies_large_fibrosity_plus_pendency} that
\begin{align}
\label{eqn_1343_1303}
        \fbr{\tau}{\GG'}+\pi_{\GG'}
        >
        \left(\frac{1}{10r\tau}-\beta+1\right)n_{\GG'}.
\end{align}
Notice that $\fbr{\tau}{\GG}=\fbr{\tau}{\GG'}$ and $\pi_\GG=\pi_{\GG'}$. Moreover, by sparsity, we have that $m_\GG=m_{\GG'}<\frac{\beta}{r-1}n_{\GG'}$. Plugging this into~\eqref{eqn_1343_1303}, using~\eqref{eqn_1618_1703}, and recalling that $\frac{\delta n}{n_\A}<m_\Y=m_\GG$, we find
\begin{align*}
        \fbr{\tau}{\GG}+\pi_{\GG}
        >
        \left(\frac{1}{10r\tau}-\beta+1\right)\frac{r-1}{\beta}m_\GG
        \geq \frac{\epsilon n_\A m_\GG}{\delta}
        >
        \epsilon n
        \geq \lfloor\epsilon n\rfloor.
\end{align*}

Let $\Omega$ be the union of the set of pendent hyperedges of $\GG$ and a maximum set of mutually disjoint $\tau$-fibers of $\GG$. Label the elements of $\Omega$ as $\omega_1,\dots,\omega_\ell$, where $\ell=\fbr{\tau}{\GG}+\pi_{\GG}$. For $i\in[\ell]$, let $J_i$ be a joint of $\omega_i$. Furthermore, consider the set $U_i$ given by $U_i=\omega_i\setminus J_i$ if $\omega_i$ is a pendent hyperedge, and $U_i=\bigcup_{e\in\omega_i}e\setminus J_i$ if $\omega_i$ is a fiber. Since at least one of $r\geq 3$ and $\tau\geq 2$ holds, the sets $U_i$ are nonempty. Moreover, they are mutually disjoint.
By the minimality of $\Y$, for each $i\in [\ell]$, $f$ is consistent with the substructure $\Y_i$ of $\Y$ induced by $Y\setminus U_i$. 
Thus, let $h_i:\Y_i\to\A$ be a homomorphism that agrees with $f$ on
$W\cap Y_i$.
Via Proposition~\ref{prop_extension_tau_fibers} or Proposition~\ref{prop_extension_pendent_edges}, we can extend each $h_i$ to a homomorphism $h'_i:\Y\to\A$.
Since $f$ is not consistent with $\Y$, for each $i\in [\ell]$ there exists some vertex $u_i\in W\cap U_i$ such that $f(u_i)\neq h_i'(u_i)$. In particular, this means that $n_\W\geq\ell$, as the sets $U_i$ are mutually disjoint.
We deduce that
\begin{align*}
    n_\W
    \geq
    \ell
    =
    \fbr{\tau}{\GG}
    +\pi_\GG
    >
    \lfloor\epsilon n\rfloor,
\end{align*}
a contradiction. 
It follows that the pair $(\X,\A)$ has a $(\lfloor\epsilon n\rfloor,\delta n)$-consistency gap, as claimed.
Observe now that
\begin{align*}
    \left(\frac{1}{10r\tau}-\beta+1\right)\frac{r-1}{\beta}
    =
    \frac{1}{10\tau}\cdot\frac{r-1}{r}\cdot\frac{1-10r\tau(\beta-1)}{\beta}
    \leq 1.
\end{align*}
Hence, using~\eqref{eqn_1618_1703} and the assumption that $n\geq\frac{1}{\epsilon}$, we deduce that
\begin{align*}
    n_\A\leq \epsilon\, n\, n_\A
    \leq
    \left(\frac{1}{10r\tau}-\beta+1\right)\frac{(r-1)\delta}{\beta}n
    \leq
    \delta n.
\end{align*}
It follows from Theorem~\ref{thm_consistent_means_bounded_width} that $\Lcon^{\lfloor\epsilon n\rfloor}(\X,\A)=\YES$. Finally, 
from $\chr(\HH)\geq c+1$ we deduce that $\X\not\to\K_{r,c}$ and, thus, 
$\X\not\to\B$. The proof is concluded. 
\end{proof}

\appendix
\section*{Appendix}
We prove the following result.
\begin{lem*}
    [Lemma~\ref{lem_enough_to_consider_monic} restated]
    Let $\A\to\B$ be relational structures, and suppose that $\PCSP(\Amon,\Bmon)$ has linear width. Then, $\PCSP(\A,\B)$ has linear width, too. 
\end{lem*}

We shall make use of a result from~\cite{BBKO21} based on the algebraic approach to PCSPs. First, we introduce the necessary terminology. Let $\sigma$ be the common signature of $\A$ and $\B$. For $d\in\N$, let $\A^d$ be the \emph{$d$-th direct power} of $\A$; i.e., $\A^d$ is the relational structure on the signature $\sigma$ whose domain is $A^d$ and whose relations are defined as follows: For $R\in\sigma$ and any $d\times\ar(R)$ matrix $M$ whose rows are tuples in $R^\A$, the columns of $M$ form a tuple in $R^{\A^d}$. We let $\Pol^{(d)}(\A,\B)$ be the set of homomorphisms $\A^d\to\B$, and we let $\Pol(\A,\B)$ (the \emph{polymorphism set} of $\PCSP(\A,\B)$) be the disjoint union of $\Pol^{(d)}(\A,\B)$ for $d\in\N$.
Given $d,d'\in\N$, $f\in\Pol^{(d)}(\A,\B)$, and a function $\pi:[d]\to[d']$, we define $f_{/\pi}$ (the ``minor of $f$ under $\pi$'') as the function from $A^{d'}$ to $B$ given by $(a_{1},\dots,a_{d'})\mapsto f(a_{\pi(1)},\dots,a_{\pi(d)})$. It is not hard to check that $f_{/\pi}$ yields a homomorphism $\A^{d'}\to\B$ and, thus, $f_{/\pi}\in\Pol^{(d')}(\A,\B)$.
For $\A\to\B$ and $\A'\to\B'$, a map $\xi:\Pol(\A,\B)\to\Pol(\A',\B')$ is a \emph{minion homomorphism}\footnote{A \emph{minion} is the algebraic structure whose operations are all minor maps.} if it preserves arities (i.e., the range of the restriction of $\xi$ to $\Pol^{(d)}(\A,\B)$ is included in $\Pol^{(d)}(\A',\B')$) and it preserves minors (i.e., $\xi(f_{/\pi})=\xi(f)_{/\pi}$ for each compatible $f$ and $\pi$). Observe that the signature $\sigma$ of $\A$ and $\B$ and the signature $\sigma'$ of $\A'$ and $\B'$ are not required to be equal for a minion homomorphism $\xi$ to exist.

It was shown in~\cite{BBKO21} that the property of having bounded width is preserved under minion homomorphisms.
The same proof also works in the linear-width regime (as was noted in~\cite{Atserias22:soda} for the case of sublinear width) and it yields the following statement.

\begin{lem}[\cite{BBKO21}]
\label{minion_homo_preserves_width}
    Let $\A\to\B$ and $\A'\to\B'$ be relational structures such that there exists a minion homomorphism $\Pol(\A,\B)\to\Pol(\A',\B')$, and suppose that $\PCSP(\A',\B')$ has linear width. Then, $\PCSP(\A,\B)$ has linear width, too.
\end{lem}
\noindent Lemma~\ref{lem_enough_to_consider_monic} follows by comparing the polymorphisms of $\PCSP(\Amon,\Bmon)$ to those of
$\PCSP(\A,\B)$.
\begin{proof}[Proof of Lemma~\ref{lem_enough_to_consider_monic}]
    Observe that any homomorphism $f$ from $\A$ to $\B$ is also a homomorphism from $\Amon$ to $\Bmon$, as easily follows from the definition of $\Amon$ and $\Bmon$. Observe also that the structures $(\A^d)^{\operatorname{mon}}$ and $(\Amon)^d$ are homomorphically equivalent for each $d\in\N$. Take a homomorphism $h_d:(\Amon)^d\to (\A^d)^{\operatorname{mon}}$. We deduce that $f\circ h_d\in \Pol^{(d)}(\Amon,\Bmon)$ for each $f\in\Pol^{(d)}(\A,\B)$. Moreover, it is immediate to check that the assignment $f\mapsto f\circ h_d$ preserves arities and minors, and it thus yields a minion homomorphism from $\Pol(\A,\B)$ to $\Pol(\Amon,\Bmon)$. The result then follows from Lemma~\ref{minion_homo_preserves_width}.
\end{proof}
We point out that the converse of Lemma~\ref{lem_enough_to_consider_monic} also holds, provided that all relations of $\A$ are nonempty. Indeed, in that case, any homomorphism from $\Amon$ to $\Bmon$ is also a homomorphism from $\A$ to $\B$, and a similar argument as in the proof above shows that there exists a minion homomorphism from $\Pol(\Amon,\Bmon)$ to $\Pol(\A,\B)$.

{\small
\bibliographystyle{plainurl}
\bibliography{cz_bibliography}
}

\end{document}